\documentclass[11pt]{article}

\usepackage{graphicx}
\usepackage{latexsym}
\usepackage{amssymb}
\usepackage{amsmath}
\usepackage{amsthm}
\usepackage{bm}
\usepackage{fullpage}
\usepackage{mathtools,color}

\usepackage{nth}

\usepackage{forloop}
\usepackage{paralist}
\usepackage{nicefrac}
\usepackage{rotating}
\usepackage{array,multirow}
\usepackage{threeparttable}
\usepackage{hyperref}

\newtheorem{thm}{Theorem}
\newtheorem{cor}[thm]{Corollary}

\newtheorem{lemma}[thm]{Lemma}
\newtheorem*{defn*}{Definition}

\newtheorem{claim}[thm]{Claim}
\newtheorem{obs}[thm]{Observation}
\newtheorem{rem}[thm]{Remark}

\newcommand\E{\mathbb{E}}
\newcommand\R{\mathbb{R}}
\newcommand\pr{\mathrm{Pr}}
\newcommand\disc{\mathrm{disc}}

\date{\vspace{-5ex}}

\newcounter{note}[section]

\makeatletter
\newtheorem*{rep@theorem}{\rep@title}
\newcommand{\newreptheorem}[2]{%
\newenvironment{rep#1}[1]{%
 \def\rep@title{#2 \ref{##1}}%
 \begin{rep@theorem}}%
 {\end{rep@theorem}}}
\makeatother

\newreptheorem{theorem}{Theorem}

\begin{document}
\setcounter{page}{0}
\title{Algorithmic Discrepancy Beyond Partial Coloring}

\author{Nikhil Bansal
\thanks{Department of Mathematics and Computer Science, Eindhoven University of Technology, Netherlands.  
Email:
\href{mailto:n.bansal@tue.nl}{n.bansal@tue.nl}.
Supported by a NWO Vidi grant 639.022.211 and an ERC consolidator grant 617951.}
\and
Shashwat Garg\thanks{Department of Mathematics and Computer Science, Eindhoven University of Technology, Netherlands.  
Email:
\href{mailto:s.garg@tue.nl}{s.garg@tue.nl}.
Supported by the Netherlands Organisation for Scientific Research (NWO) under project no.~022.005.025.
}
}

\maketitle

\begin{abstract}
The partial coloring method is one of the most powerful and widely used method in combinatorial discrepancy problems.
However, in many cases it leads to sub-optimal bounds as the partial coloring step must be iterated a logarithmic number of times, and the errors can add up in an adversarial way.

We give a new and general algorithmic framework that overcomes the limitations of the partial coloring method and can be applied in a black-box manner to various problems. 
Using this framework, we give new improved bounds and algorithms for several classic problems in discrepancy. In particular, for Tusnady's problem, we give an improved $O(\log^2 n)$ bound
for discrepancy of axis-parallel rectangles and more generally an $O_d(\log^dn)$ bound for $d$-dimensional boxes in $\mathbb{R}^d$. Previously, even non-constructively, the best bounds were $O(\log^{2.5} n)$ and $O_d(\log^{d+0.5}n)$ respectively.
Similarly, for the Steinitz problem we give the first algorithm that matches the best known non-constructive bounds due to 
Banaszczyk \cite{Bana12} in the $\ell_\infty$ case, and improves the previous algorithmic bounds substantially in the $\ell_2$ case.
Our framework is based upon a substantial generalization of the techniques developed recently in the context of the Koml\'{o}s  discrepancy problem \cite{BDG16}.


\end{abstract}

\clearpage

\section{Introduction}

Let $(V,\mathcal{S})$ be a finite set system, with $V=\{1,\ldots,n\}$ and $\mathcal{S} = \{S_1,\ldots,S_m\}$ a collection of subsets of $V$. For a two-coloring $\chi: V  \rightarrow \{-1,1\}$, the discrepancy of $\chi$ for a set $S$  is defined as $ \chi(S) =  |\sum_{j\in S} \chi(j) |$ and measures the imbalance from an even-split for $S$.
The discrepancy of the system $(V,\mathcal{S})$ is defined as 
\[ \disc(\mathcal{S}) = \min_{\chi:V \rightarrow \{-1,1\}} \max_{S \in \mathcal{S}} \chi(S). \]
That is, it is the minimum imbalance for all sets in $\mathcal{S}$, over all possible two-colorings $\chi$.
More generally for any matrix $A$, its discrepancy is defined as $\disc(A) = \min_{x \in \{-1,1\}^n} \|Ax\|_\infty$.

Discrepancy is a widely studied topic and has applications to many areas in mathematics and computer science.
In particular in computer science, it arises naturally in computational geometry, data structure lower bounds, rounding in approximation algorithms, combinatorial optimization, communication complexity and pseudorandomness. 
For much more on these connections we refer the reader to the books \cite{Chazelle,Mat09,Panorama}. 

\paragraph{Partial Coloring Method:} One of the most important and widely used technique in discrepancy is the partial coloring method developed in the early 80's  by Beck, and its refinement by Spencer to the entropy method \cite{Beck81b, Spencer85}. An essentially similar approach, but based on ideas from convex geometry was developed independently by Gluskin \cite{Glu89}. 
Besides being powerful, an important reason for its success is that it can be applied easily to many problems in a black-box manner and for most problems in discrepancy the best known bounds are achieved using this method.
While these original arguments were based on the pigeonhole principle and were non-algorithmic,
in recent years several new algorithmic versions of the partial coloring method have been developed \cite{B10,LM12,Ro14,HSS14,ES14}.
In particular, all known applications of partial-coloring \cite{Spencer85,Mat09} can now be made algorithmic. These ideas have also led to several other new results in approximation algorithms \cite{R13,BCKL14,BN15,NTZ13}.

In many applications however, the partial coloring method gives sub-optimal bounds. The problem is that
this method finds a low discrepancy coloring while coloring only a constant fraction of the elements, and it
must be iterated logarithmically many times to get a full coloring.
The iterations are unrelated to each other and the error can add up adversarially over the iterations.
A well known example that illustrates this issue is the problem of understanding the discrepancy of 
sparse set systems where each element lies in at most $t$ sets.
Here each partial coloring step incurs $O(t^{1/2})$ discrepancy and the overall discrepancy becomes $O(t^{1/2} \log n)$. On the other hand, the celebrated Beck-Fiala conjecture is that the (overall) discrepancy must be $O(t^{1/2})$.
Similar gaps exist for several other classic problems.
For example, for Tusnady's problem about discrepancy of points and axis aligned rectangles (details in Section \ref{sec:applications}), partial coloring gives an upper bound of $O(\log^{2.5} n)$ while the best known lower bound is $\Omega(\log n)$. Similarly, for the well known Steinitz problem and its several variants (details in Section \ref{sec:applications}), partial coloring gives an $O(d^{1/2} \log n)$ bound while the conjectured answer is $O(d^{1/2})$.

\paragraph{Banaszczyk's approach:} In a breakthrough result, Banaszczyk \cite{B97} used deep techniques from convex geometry to bypass the partial coloring barrier and gave an $O(t^{1/2} \log^{1/2} n)$ discrepancy bound for the Beck-Fiala problem (and the more general Koml\'{o}s problem). In particular, he gave a general result that given any collection of vectors of $\ell_2$ norm at most one and any convex body $K$ with Gaussian volume $1/2$, there exists a $\pm 1$ signed combination of the vectors that lies in $cK$ for some constant $c$. 
In recent years, several remarkable applications of this result have been found.
Banaszczyk \cite{Bana12} used it to obtain improved bounds for the Steinitz problem,  
\cite{MNT14} used it to relate the $\gamma_2$-norm and hereditary discrepancy and to get an approximation algorithm for hereditary discrepancy, and \cite{Lar14} used it to give space-query tradeoffs for dynamic data structures. While we do not know how to make a formal connection to the partial coloring method, roughly speaking, Banaszczyk's approach allows the errors during each partial coloring to accumulate in an $\ell_2$ manner, instead of in an $\ell_1$ manner.


However, Banaszczyk's original proof \cite{B97} is rather deep and mysterious and does not give any efficient algorithm for finding a good coloring. Finding an algorithmic version of it is a major current challenge \cite{R-convex14, Sashothesis, DGLN16}. Recently, the authors together with Daniel Dadush \cite{BDG16} gave an efficient algorithm for the Koml\'{o}s problem matching Banaszczyk's bound.
The key idea here was to use an SDP with several additional constraints, compared to the earlier SDP approach of \cite{B10}, so that the associated random walk satisfies some extra properties. Then, a more sophisticated martingale analysis based on Freedman's inequality was used to bound the deviation of the discrepancy from the expected value in the random walk. 

Despite the progress, there are several limitations of this result. First, the SDP was specifically tailored to the Koml\'{o}s problem and rather adhoc. Second, the analysis was quite technical and specific to the Koml\'{o}s problem. More importantly, it does not give any general black-box approach like those given by the partial coloring method and its algorithmic variant due to Lovett and Meka \cite{LM12}, or Banaszczyk's original result, which can be applied directly to a problem without any understanding of the underlying algorithm or its proof.

%
%
 %



\section{Our Results}
We give a new general framework that overcomes the limitations of the partial coloring method, and gives improved algorithmic bounds for several problems in discrepancy. Moreover, it can be applied in a black-box manner to any discrepancy problem without the need to know any inner workings of the algorithm or its proof of correctness. Below we give an informal description of the framework and the main result, and defer the formal version to Section \ref{sec:framework} until some necessary notation is developed.

\subsection{Framework}

The framework is best viewed as a game between a player and the algorithm. Let $B$ be an $m\times n$ matrix with entries $\{a_{ji}\}$.
At each time step $t$, the player can specify some subset of uncolored elements $A(t)$ to be colored. 
Moreover, the player can specify up to $\delta |A(t)|$ linear constraints (where $\delta <1$) of the type 
\[\sum_{i \in A(t)} w_{k}(i) \Delta x_i(t)=0 \qquad \text{ for } k=1,\ldots,\delta |A(t)|\]
where $\Delta x_i(t)$ denotes the color update of element $i$ at time $t$. 
 Let us call these constraints $Z(t)$. The algorithm then updates the colors subject to the constraints $Z(t)$.
This game continues repeatedly until all elements are colored $\pm 1$.

Fix some row $j$ of $B$ and any subset $S$ of elements. Based on what $Z(t)$'s are picked during the process and how they relate to $S$ and $j$, determines whether an element $i \in S$ is ``corrupted" or not (with respect to $j$ and $S$). We defer the description of how an element gets corrupted to Section~\ref{sec:framework}. At the end of the process, the final coloring satisfies the following guarantee.



\begin{thm}\label{thm:vecgenunified}(informal version)
There is a constant $c>0$ such that given an $m\times n$ matrix $B$ with entries $\{a_{ji}\}$ satisfying $|a_{ji}|\le 1$, then for any row $j$, any subset $S \subseteq [n]$ of elements and $\lambda\ge 0$, the coloring $\chi \in \{-1,1\}^n$ returned by the algorithm satisfies,
\[ \pr \left[ \left|\sum_{i \in S} a_{ji} \chi(i)\right| \geq c \lambda \left( (\sum_{i \in C_{j,S}} a_{ji}^2 )^{1/2} + \lambda \right)  \right] \leq 2 \exp(-\lambda^2/2) \] 
where $C_{j,S}$ denotes the set of corrupted elements with respect to $j$ and $S$.

\end{thm}

Observe that the discrepancy behaves as a sub-gaussian  with standard deviation at most the $\ell_2$ norm of the corrupted elements (provided $\lambda \leq  (\sum_{i \in C_{j,S}} a_{ji}^2 )^{1/2}$, later we will show that this restriction is necessary). Moreover, the bound holds for the full coloring returned by the algorithm, as opposed to a partial coloring.

Also note that this discrepancy bound holds for any subset $S$, unlike other previous approaches that usually only give low discrepancy for specific sets $S$ corresponding to rows of 
$B$. In particular, due to arbitrary correlations between the colors of various elements, all the previous approaches for discrepancy that we know of are incapable of giving such a guarantee.  
In contrast, our algorithm returns a probability distribution over colorings in $\{-1,1\}^n$ which gives an almost sub-gaussian tail bound on the discrepancy for every row $j$ and subset $S\subseteq [n]$ simultaneously.

%
%
%

\subsection{Applications}
\label{sec:applications}
We apply this framework to obtain several new algorithmic results for various problems in combinatorial discrepancy. In fact, all these results follow quite easily
by choosing $A(t)$ and $Z(t)$ in a natural way, as determined by the structure of the problem at hand.

\paragraph{Tusnady's Problem:} 
Given a set $P$ of $n$ points in $\mathbb{R}^d$, let $disc(P,\mathcal{R}^d)$ denote the discrepancy of the set system with $P$ as the elements and sets consisting of all axis-parallel boxes in $\mathbb{R}^d$.
Understanding the discrepancy of point sets with respect to axis-parallel boxes has been studied in various different settings since the origins of discrepancy theory in the 1930's and has a fascinating history, see e.g.~\cite{Chazelle,Mat99,Panorama}. Determining the correct order of magnitude for 
\[disc(n,\mathcal{R}^d) =\sup_{|P|=n}disc(P,\mathcal{R}^d)\]
 has received a lot of attention \cite{Beck81, Beck89, Boh90, Srin97, Lar14, Mat99, MNT14} and the best known bounds prior to this work were $\Omega_d(\log^{d-1}n)$ \cite{MNT14} and $O_d(\log^{d+1/2}n)$\cite{Lar14, Mat99}. 
Here we use $O_d(.)$ notation to hide factors depending only on $d$.
In particular, even for $d=2$, there was a relatively large gap of $\Omega(\log n)$ and $O(\log^{2.5} n)$ between the lower and upper bound. 
 
Using the framework above, we can show the following improved upper bound.
\begin{thm}
\label{thm:tusnady}
Given any set $P$ of $n$ points in $\mathbb{R}^d$, there is an efficient randomized algorithm to find a $\{-1, 1\}^n$ coloring of $P$ such that the discrepancy of all axis-parallel boxes is $O_d(\log^d n)$.
\end{thm}
Interestingly, nothing better than $O_d(\log^{d+1/2} n)$ was known even non-constructively, even for $d=2$, prior to our work.
After our result was announced Nikolov\cite{Sasho16} further improved this discrepancy bound to $O_d(\log^{d-1/2}n)$ using a clever application of Banaszczyk's results \cite{B97, Bana12}, coming tantalizingly close to the $\Omega(\log^{d-1})$ lower bound \cite{MNT14}. However, this result is not algorithmic. \\

\emph{Discrepancy of Polytopes:} We also extend the above theorem to discrepancy of polytopes generated by a fixed set of $k$ hyperplanes. Let $H$ be a set of $k$ hyperplanes in $\mathbb{R}^d$ and let $POL(H)$ denote the set of all polytopes of the form $\cap_{i=1}^{\ell} h_i $ where each $h_i$ is a halfspace which is a translation of some halfspace in $H$. We wish to color a given set of $n$ points such that the discrepancy of every polytope in $POL(H)$ is small.
This problem was considered in \cite{Mat99, Mat09} where a bound of $O_{d,k}(\log^{d+1/2}\sqrt{\log\log n})$ was given, which prior to our work was the best known even for $d=2$. We improve this bound in the following theorem.
\begin{thm}
\label{thm:polytopes}
Given a set $P$ of $n$ points in $\mathbb{R}^d$, there exists an efficient randomized algorithm to find a $\{-1,1\}^n$ coloring of $P$ such that the discrepancy of all polytopes in $POL(H)$ is $O_{d,k}(\log^dn)$.
\end{thm}

\paragraph {Steinitz Problem:} Given a norm $\|.\|$ and any set $V=\{v_1,\ldots,v_n\}$ of $n$ vectors in $\mathbb{R}^d$ with  $\|v_i\| \le 1$ for each\footnote{We use $[n]$ to denote the set \{1,2,\dots,n\}.} $i\in [n]$ and satisfying $\sum_{i=1}^n v_i=0$, the Steinitz problem asks for the smallest number $B$, depending only on $d$ and $\|.\|$, for which there exists an ordering $v_{\pi(1)},v_{\pi(2)},\dots,v_{\pi(n)}$ of these vectors such that all partial sums along this ordering have norm $\|.\|$ bounded by $B$ i.e. 
\[ \|\sum_{i=1}^k v_{\pi(i)}\|\le B \qquad \textrm{for all } 1\le k\le n .\]
This question was originally asked by Riemann in early 1900's and has several applications in areas such as graph theory \cite{AB86}, integer programming \cite{kevin12,Dash12} and scheduling \cite{Sev}. For an interesting history and other applications, we refer to the excellent surveys \cite{Ando89,bar,Sev}.


Steinitz showed that $B\le 2d$ for any norm \cite{Stz16}. This was later improved to $B\le d$ \cite{GS90}.
The special cases of $\ell_2$ and $\ell_\infty$ norms have been widely studied, and while no 
$o(d)$ bounds are known, it is believed that $O(\sqrt{d})$ should be the right answer.
Recently,
Banaszczyk \cite{Bana12} showed (non-constructive) bounds of $O(\sqrt{d\log n})$ for the $\ell_\infty$ norm and $O(\sqrt{d}+\sqrt{\log n})$ for the $\ell_2$ norm,
based on a non-trivial application of his result in \cite{B97}. 
On the algorithmic side, the bound of $d$ due to \cite{GS90} is based on a clever iterated rounding procedure.
Recently, 
\cite{Har14} considered algorithms for the $\ell_2$ case, motivated by a central question in machine learning that is also referred to as the {\em Herding Problem}, and gave an $O(\sqrt{d} \log^{2.5} n)$ bound.

A closely related, and harder problem, is 
the \textit{signed series problem}. Here given a sequence of vectors $v_1,\dots,v_n$, each with norm at most $1$, to goal is to find signs $\{-1,1\}^n$ such that the norm of signed sum for each prefix of the sequence is bounded by some number $E$, depending only on $d$ and the norm. Chobanyan \cite{Chob94} gave a general reduction from the Steinitz problem to the signed series problem and showed that $B\le E$.
The results of \cite{Bana12,Har14} in fact hold for this harder problem and \cite{Har14} showed that this reduction can be made algorithmic.

Using our framework, we give the following bound for the $\ell_\infty$ case of signed series problem. 
\begin{thm}
\label{thm:stff}
Given vectors $v_1,\dots,v_n\in\mathbb{R}^d$ with $\|v_i\|_\infty\le 1$ for all $i\in[n]$, there is an efficient randomized algorithm to find a coloring $\chi \in \{-1,1\}^n$ such that $\max_{k\in[n]}\|\sum_{i=1}^k \chi(i) v_{i}\|_\infty \le O(\sqrt{d\log n})$. 
\end{thm}
This matches Banaszczyk's \cite{Bana12} bound for the signed series problem, and by the reduction of \cite{Chob94,Har14} also implies an $O(\sqrt{d\log n})$ algorithmic bound for the Steinitz problem.



For the $\ell_2$ case, we give an algorithmic $O(\sqrt{d\log n})$ bound for the signed series problem (and hence also for Steinitz problem).
While this does not match Banaszczyk's bound \cite{Bana12} of $O(\sqrt{d}+\sqrt{\log n})$, it improves upon the algorithmic $O(\sqrt{d}\log^{2.5}n)$ bound in \cite{Har14}.

\begin{thm}
\label{thm:st22}
Given vectors $v_1,\dots,v_n\in\mathbb{R}^d$ with $\|v_i\|_2\le 1$ for all $i\in[n]$, there is an efficient randomized algorithm to find a coloring $\chi \in \{-1,1\}^n$ such that $\max_{k\in[n]}\|\sum_{i=1}^k \chi(i) v_{i}\|_2 \le O(\sqrt{d\log n})$. 
%
\end{thm}
To show this we use our framework, but modify the analysis in the proof of Theorem \ref{thm:vecgenunified} slightly to adapt to $\ell_2$ discrepancy.

%

\paragraph{Koml\'{o}s Problem:} 
Given a collection of vectors $v_1,\ldots,v_n \in R^d$ of $\ell_2$ norm at most one, for some arbitrary $d$, the goal here is to find a $\pm 1$ coloring $\chi$  of these vectors that minimizes $\|\sum_i \chi(i) v_i\|_\infty$. The celebrated Koml\'{o}s conjecture says that an $O(1)$ discrepancy coloring always exists. The best known bound due to Banaszczyk is $O(\sqrt{\log n})$ \cite{B97} and this  was recently matched algorithmically in \cite{BDG16}. 
Our framework here directly gives an $O(\sqrt{\log n})$ bound, with a much cleaner proof than in \cite{BDG16}.


%
%

\subsection{Universal Vector Colorings}
One of the key technical ingredients behind Theorem~\ref{thm:vecgenunified} is the existence of certain vector (partial) colorings satisfying very strong properties.
Vector colorings are a relaxation of $\pm 1$ colorings where the color of each point is allowed to be a vector of length at most $1$ and can be found by writing a semi-definite program (SDP).
For reasons explained below, we call the vector colorings we use as {\em Universal Vector Colorings}. These vector colorings can also be viewed as the $\ell_2$-analogues of 
{\em Basic Feasible Solutions}, that play a crucial role and are widely studied in LP rounding algorithms \cite{LRSbook}. 
We elaborate on this connection further below and believe that these colorings should have new applications in rounding fractional solutions in approximation algorithms. 

\begin{thm}(Universal Vector Colorings)
\label{thm:univec}
Let $[n]$ be a set of elements. Let $0 \leq \delta < 1$ and fix $\beta\in(0,1-\delta)$.  Given an arbitrary collection of $\ell$ linear constraints specified by vectors $w_1,w_2,\dots,w_{\ell} \in\mathbb{R}^n$ where $\ell=\delta n$, there exists a collection of vectors $\{u_i\}_{i=1}^n$ such that:
\begin{enumerate}[i)]
\item The vector discrepancy along each direction $w_k$ is zero i.e.,
\[ \sum_{i}w_k(i)u_i=0 \qquad \textrm{ for all }1\le k\le \ell.\] 
\item For every vector $b\in\mathbb{R}^n$, it holds that 
\[\|\sum_{i}b(i)u_i\|_2^2\le \frac{1}{\beta}\sum_i b(i)^2 \|u_i\|_2^2.\]
\item For each $i \in [n]$ we have $\|u_i\|_2\le 1$, and $\sum_i \|u_i\|_2^2\ge (1-\delta-\beta)n$.
\end{enumerate}
Moreover, the $\{u_i\}$'s can be computed in polynomial time by solving an SDP.
\end{thm}

In particular, the first property requires that the vectors $u_i$ be nicely correlated so that their weighted sum is $0$ in each of the $\ell$ directions given by $w_k$'s. On the other hand, the second property requires that these vectors be almost orthogonal in a very strong sense. In particular, the property is satisfied for $\beta=1$ iff $u_i$ are mutually orthogonal.  
One trivial way to satisfy these conditions is to pick $u_i=0$ for all $i$. But the third property states that most of these vectors have length $\Theta(1)$.
 
%
%
%
%

\paragraph{Relation to Basic Feasible Solutions for LPs.}
Consider a linear program on $n$ variables $x_1,\ldots,x_n$, with constraints $x_i \in [-1,1]$ for all $i \in [n]$ and 
$w_k \cdot x =0$ for $k=1,\ldots,\ell$. 
Then any basic feasible solution to the LP satisfies that at least $n-\ell$ variables are set to $-1$ or $1$ (and hence $\sum_{i\in [n]} |x_i| \geq n-\ell$).
Moreover, for any vector $b \in \mathbb{R}^n$, it trivially holds that $b \cdot x \leq \|b\|_1$. These simple properties are 
the key to the iterated rounding technique \cite{LRSbook}.

The second and third properties in Theorem \ref{thm:univec} can be viewed as analogous to these. 
While the property $b \cdot x \leq \|b\|_1$ follows trivially, it is instructive to consider the following example which shows why the $\ell_2$ analogue is nontrivial (this is also the reason why the error adds up in an $\ell_1$ manner during the iterations of the partial coloring). 
Suppose we have $\ell=n/2$ constraints $u_1 =u_{\ell+1}$,  $u_2=u_{\ell+1}$, \ldots, $u_{\ell}=u_{\ell+1}$. This would enforce that 
$u_1=\ldots=u_{\ell+1}$. For $\beta = 1/4$, we require that for every $b$, $ \|\sum_i b(i)u_i \|^2 \leq 4 \sum_i b(i)^2 \|u_i\|^2$. However, choosing $b$ to be (say) the vector with $1$ in the first $n/2$ coordinates, and $0$ elsewhere, we would have
that $\|\sum_i b(i)u_i \|^2 =  \|\sum_{i=1}^{n/2} u_i \|^2  = (n/2)^2 \|u_1\|^2$, which is substantially larger than $
4 \sum_i b(i)^2 \|u_i\|^2 = 2n \|u_1\|^2$ unless $u_1=\ldots=u_{n/2}=0$.
In particular, the first and second requirements in Theorem \ref{thm:univec} can interact in complicated ways and it is not trivial to still guarantee the third property. Using Theorem \ref{thm:univec} we can get the property that $b\cdot x=O(\|b\|_2)$ for all $b\in \mathbb{R}^n$ and thus the additive error incurred while rounding a fractional solution can be much smaller than the error of $\|b\|_1$ incurred in iterated rounding.



\paragraph{Comparison to Beck's partial coloring lemma.}
Interestingly, the previous approaches for discrepancy such as partial coloring lemma, or 
the SDP based algorithms such as \cite{B10,BDG16} can be viewed as enforcing the second property only for very specific choices of vectors $b$, tailored to the specific problem at hand.
For concreteness, let us consider Beck's partial coloring lemma. 
\begin{lemma}\cite{Beck81b}
\label{thm:beck}
Let $\mathcal{F}$ and $\mathcal{M}$ be set systems on an $n$ point set $V$ such that $|M|\le s$ for every $M\in\mathcal{M}$ and 
\[ \prod_{F\in\mathcal{F}}(|F|+1)\le 2^{(n-1)/5}.\]
Then there exists a partial coloring $\chi:V\rightarrow\{-1,0,1\}$, such that at least $n/10$ elements of $V$ are colored, $\chi(F)=0$ for every $F\in\mathcal{F}$ and $|\chi(M)|\le O(\sqrt{s\log |\mathcal{M}|})$ for every $M\in\mathcal{M}$.
\end{lemma}

That is, it gives a partial coloring with zero discrepancy on some $\delta' n$ sets for $\delta' \ll 1$, and 
guarantees an essentially $\sqrt{s}$ discrepancy (ignoring the $\sqrt{\log}$ factor) for the remaining sets in the system.
This is the second property in Theorem \ref{thm:univec}, for $b$ corresponding to the indicator vectors of these sets.
  

The fact that the vector colorings in Theorem \ref{thm:univec} satisfy the second property for {\em every} $b$, is the reason why we call them as Universal Vector Colorings, and this property will play a crucial role in the design of our framework.



\subsection{Organisation of the Paper}
We present the formal statement of our framework and Theorem~\ref{thm:vecgenunified} in Section~\ref{sec:framework}. The main ingredient used in the proof of the framework is Theorem~\ref{thm:univec} concerning Universal Vector Colorings which is proved in Section~\ref{sec:univec}. 
The rest of the proof of Theorem~\ref{thm:vecgenunified} is deferred to Appendix A to keep the focus more on how to apply our framework in a black-box manner to problems in discrepancy.
 
Theorem~\ref{thm:tusnady} concerning Tusnady's problem is proved in Section~\ref{sec:tusnady}.  
Theorem~\ref{thm:polytopes} follows by combining the idea in Theorem~\ref{thm:tusnady} and using the decomposition of polytopes in $POL(H)$ into simpler shapes given in \cite{Mat99}. Due to this, we only provide a sketch of the proof in Appendix B.
Theorem~\ref{thm:stff} and Theorem~\ref{thm:st22} concerning Steinitz problem are proved in Section~\ref{sec:steinitz}. 

\section{Framework}
\label{sec:framework}


We now describe the framework formally. 
Let $B$ be a $m\times n$ matrix with entries $\{a_{ji}\}$. Usually in discrepancy, given a coloring $x\in \{-1,1\}^n$  
one considers the maximum discrepancy $\|Bx\|_\infty$ over each row. However in our setting, it will be more convenient to look at 
the discrepancy  for any row $j$ and any subset $S \subseteq [n]$ of elements. This for example will allow us to directly argue about prefixes of rows of $B$ in the Steinitz problem.
For a coloring $x$, let \[\textrm{disc}(x,j,S) = \sum_{i\in S} a_{ji}  x(i)\] denote this discrepancy for the row-subset pair $(j,S)$.

The overall algorithm will proceed in various time steps (rounds).
As in previous algorithms for discrepancy, we work with fractional colorings where the color of an element $i$ at step $t$ is $x_t(i) \in[-1,1]$. Initially all elements are colored $x_0(i)=0$. Eventually all the variables will reach either $-1$ or $1$.
We call variable $i$ alive at time $t$ if $|x_t(i)|< 1-1/n$, 
and once $|x_t(i)| \ge 1-1/n$, we call this variable frozen and its value is not updated any more.
Let $N(t)=\{i\in[n] : |x_{t-1}(i)| < 1-\frac{1}{n}\}$ denote the set of alive variables at the beginning of time $t$.


\paragraph{Game View:} It is useful to view the framework as a game between a player and a black-box algorithm (not necessarily adversarial).
The goal of the player is to get a low discrepancy coloring at the end. 
At each step $t$,
the player can choose a subset $A(t) \subseteq N(t)$ of elements whose colors are allowed to be updated and also specifies a collection of linear constraints $Z(t)$ on the variables in $A(t)$. Let $w_1,\ldots,w_{\ell} \in \mathbb{R}^{|A(t)|}$ denote the vectors specifying these constraints, where $\ell \leq \delta |A(t)|$ for any fixed constant $\delta <1$.
 
 
The black-box produces some updates $\Delta x_t(i)$ that satisfy the constraints 
\[\sum_{i \in A(t)} w_k(i) \Delta x_t(i) =0 \qquad \forall k=1,\ldots, \ell \]
 and the current coloring is updated as $x_t(i) = x_{t-1}(i) + \Delta x_t(i)$. This process repeats until every element is colored.
We remark that the black-box simply produces a Universal Vector Coloring 
 on the variables in $A(t)$ with constraints $Z(t)$ and applies the standard random projection rounding to get the updates $\Delta x_t(i)$.

For a given problem, suppose one cares about minimizing the discrepancy of certain row-subset pairs $(j,S)$.
The aim of the player will be to choose the sets $A(t)$ and $Z(t)$ adaptively at each step, so that it can ``protect" as much of the pair $(j,S)$ as possible until the of the algorithm.
The final discrepancy of $(j,S)$ will depend on how much of $(j,S)$ could not be protected as stated in Theorem \ref{thm:vecgenunified} below.

\paragraph{Protection and Corruption:} 
We now define what it means to protect an element $i$ with respect to some pair $(j,S)$. Fix some pair $(j,S)$ and fix a time step $t$. An element $i \in S$ is already protected at time $t$ if $i\notin A(t)$.
So it suffices to consider elements in $S\cap  A(t)$. We say a constraint $w_k \in Z(t)$ is {\em eligible} for $S$ if $\textrm{supp}(w_k) 
\subseteq S$, where $\textrm{supp}(w)=\{i: w(i)\neq 0\}$ is the support of $w$. For a subset $S'\subseteq S\cap A(t)$, let $v_{S',j,t}$ denote the vector defined as having $a_{ji}$ in the $i^{th}$ entry if $i\in S'$ and $0$ otherwise.

The player can pick any subset $H \subseteq [\ell]$ such that each $w_k$ for $k\in H$ is eligible for $S$. We say that a subset $S'\subseteq S \cap A(t)$ is protected at time $t$ with respect to $(j,S)$, if
\begin{equation}
\label{eqn:disjointsets}
 \sum_{k \in H} w_{k} = v_{S',j,t} .
\end{equation} 
In other words, the sum of the vectors $w_k$ for $k\in H$ is exactly the vector $v_{S',j,t}$. Any element $i\in S'$ is called {\em protected} for $(j,S)$ at time $t$. An element $i$ is called {\em corrupted} for $(j,S)$ if there was any time $t$ when $i$ was not protected with respect $(j,S)$.

\begin{rem}
\label{rem:relax}
We note that we can relax the condition \eqref{eqn:disjointsets} by letting the vector $v_{S',j,t}$ have non-zero entries on the corrupted elements in $(S \cap A(t) )\setminus S'$; for an element $i\in (S \cap A(t) )\setminus S'$, it suffices to have $| v_{S',j,t}(i) | =O(|a_{ji}|)$. 
\end{rem}

\paragraph{Example:} It is instructive to consider an example. Suppose all the entries $a_{ji}$ are $0$ or $1$ and $S$ is some set whose discrepancy we care about.
Suppose that the constraints $w_k$ also have $0$-$1$ coefficients.
Then we can pick any subset $H \subseteq [\ell]$, such that $\textrm{supp}(w_k) \subseteq S$ and $\textrm{supp}(w_k) \cap \textrm{supp}(w_{k'})=\emptyset$ for any two distinct $k,k' \in H$. An element $i \in S \cap A(t)$ is not protected if $i$ is not contained in any $w_k$ for $k\in H$.

%
%
%


As we shall see in the applications, for a given problem at hand there is usually a natural and simple way to choose the vectors $w_k$ and $H$, and apply the framework. We are now ready to state our main result.

\begin{reptheorem}{thm:vecgenunified}
There is a constant $c>0$ such that given an $m\times n$ matrix $B$ with entries $\{a_{ji}\}$ satisfying $|a_{ji}|\le 1$, then for any row $j$, subset $S \subseteq [n]$ of elements and $\lambda\ge 0$, the coloring $\chi \in \{-1,1\}^n$ returned by the above algorithm satisfies,
\[ \pr \left[ |disc(\chi,j,S)|  \geq c \lambda \left( (\sum_{i \in C_{j,S}} a_{ji}^2 )^{1/2} + \lambda \right)  \right] \leq 2 \exp(-\lambda^2/2) \] 
where $C_{j,S}$ denotes the set of corrupted elements for $(j,S)$. 
\end{reptheorem}


The above bound holds for the full coloring returned by the algorithm, whereas previous works could only ensure such a bound for a partial coloring. 

\paragraph{Application to the Beck-Fiala problem:} As an example, let us see how one can easily recover the main result of \cite{BDG16} to get an $O(\sqrt{t\log n})$ discrepancy bound in the Beck-Fiala setting. Here we have a set system $(V,\mathcal{S})$ such that each point lies in at most $t$ sets. We apply Theorem~\ref{thm:vecgenunified} and choose $A(k)=N(k)$ at all time steps $k$. We choose $Z(k)$ to contain the indicator vector of those sets in $\mathcal{S}$ which contain more than $4t$ alive points at time $k$. As each point lies in at most $t$ sets, $|Z(k)|\le |A(k)|/4$ and $\delta\le 1/4$. We will denote by $w_S$ the indicator vector of set $S$.
At any time $k$, for a set $S_j\in\mathcal{S}$, either $w_{S_j}\in Z(k)$ (and all points in $S_j$ are protected) or $S_j$ has at most $4t$ alive points. Thus there can be at most $4t$ corrupted points for $S_j$. Applying Theorem~\ref{thm:vecgenunified} now with $\lambda=O(\sqrt{\log n})$ gives a discrepancy of $O(\sqrt{t\log n})$ to $S_j$ with high probability (assuming $t>\log n$). 


Interestingly, we can also show that the above algorithm gets a good bound on the $\ell_2$ discrepancy for any subset $S\subseteq [n]$. 


\begin{thm}
\label{thm:l2discrepancy}
There is a constant $c>0$ such that given an $m\times n$ matrix $B$ having each column of length at most one ( $\sum_j a_{ji}^2\le 1$ for all $i$), then for any subset $S \subseteq [n]$ of elements and $\lambda\ge 0$, the coloring returned by the above algorithm satisfies the following bound on the $\ell_2$ discrepancy of $S$,
\[ \Pr \left[ \left(\sum_j disc(\chi,j,S)^2\right)^{1/2}  \geq c \lambda \left( (\sum_j\sum_{i \in C_{j,S}} a_{ji}^2 )^{1/2} + \lambda \right)  \right] \leq 3 \exp(-\lambda^2/2) \] 
where $C_{j,S}$ denotes the set of corrupted elements for $(j,S)$. 
\end{thm}

It should be pointed out that the above bound for $\ell_2$ discrepancy does not follow from the bound on $\ell_\infty$ discrepancy in Theorem~\ref{thm:vecgenunified}. For instance if $\lambda=\sqrt{\log n}$ but each row were to have a much smaller $\ell_2$ norm, Theorem~\ref{thm:vecgenunified} will give a discrepancy of $c\lambda^2=O(\log n)$ to each row rather than $O((\sum_{i\in C_{j,S}} a_{ji}^2\log n)^{1/2})$. This will then give a weaker bound of $O(\sqrt{m}\log n)$ on the $\ell_2$ discrepancy rather than a bound of $O((\sum_j\sum_{i\in C_{j,S}} a_{ji}^2\log n)^{1/2})$ which is always smaller than $O(\sqrt{m\log n})$. Though the same algorithm works for both Theorems, the analysis in the proof of Theorem~\ref{thm:l2discrepancy} needs to be modified to adapt to $\ell_2$ discrepancy.

\section{Universal Vector Colorings}
\label{sec:univec}

In this section, we prove Theorem~\ref{thm:univec}.
To find a vector coloring satisfying conditions (i)-(iii) in Theorem~\ref{thm:univec}, we write constraints on a PSD matrix $X$; the vectors $u_i$ will then be given by the columns of the matrix $U$ obtained by a Cholesky decomposition of $X=U^T U$ i.e. $X_{ij}=\langle u_i,u_j\rangle$. 

Condition (i) in Theorem~\ref{thm:univec} can be encoded as the following SDP constraints:
\[ w_kw_k^T\bullet X=0 \qquad \textrm{ for all }1\le k\le \ell\]
Similarly the first part of condition (iii) can be encoded as
\[ e_ie_i^T\bullet X\le 1 \qquad \textrm{ for all }i\in[n]\]
where $e_i$ is the standard unit vector in the $i$-th coordinate.

The following lemma allows us to write condition (ii) succinctly.
\begin{lemma}
Condition (ii) in Theorem~\ref{thm:univec}
\begin{equation}
\label{pr2}
 \|\sum_i b(i)u_i\|_2^2\le \frac{1}{\beta}\sum_ib(i)^2\|u_i\|_2^2 \qquad \textrm{ for all }b\in\mathbb{R}^n
 \end{equation}
is equivalent to
\begin{equation}
\label{eq:lmi}
 X\preceq \frac{1}{\beta}diag(X)
\end{equation}
where $diag(X)$ is the matrix restricted to the diagonal entries of $X$.
\end{lemma}
\begin{proof}
As $\|u_i\|^2 = X_{ii}$ and the left hand side of \eqref{pr2} is $\sum_{ij} b(i) b(j) X_{ij}$, (\ref{pr2}) can be rewritten as 
\[ b^TXb\le \left(\frac{1}{\beta}\right)b^Tdiag(X) b  \qquad \textrm{ for all }b\in\mathbb{R}^n .\]
Rearranging the terms, this equals
\[ b^T(\left(\frac{1}{\beta}\right)diag(X) -X)b\ge 0\qquad \textrm{ for all }b\in\mathbb{R}^n\]
which is equivalent to the matrix $\frac{1}{\beta}diag(X)-X$ being PSD. 
\end{proof}
Notice that \eqref{eq:lmi} is a Linear Matrix Inequality (LMI) and hence a valid SDP constraint. Finally, we write the SDP as

\begin{eqnarray}
	\textrm{Maximize }tr(X) \nonumber\\
	\textrm{s.t.} \qquad 
	\label{sdpuv1} w_kw_k^T\bullet X &=&0 \qquad \textrm{ for all }1\le k\le \ell \\
	\label{sdpuv2} X &\preceq &\frac{1}{\beta}diag(X)\\
   	\label{sdpuv3} e_ie_i^T\bullet X & \le &  1  \qquad \forall i \in [n]  \\
	\label{sdpuv4} X &\succeq &0 \nonumber
\end{eqnarray}
The objective function of this SDP is to maximize the trace of $X$ which is the same as maximizing $\sum_i \|u_i\|_2^2$. The SDP is feasible as $X=0$ satisfies all the constraints. Hence to prove Theorem~\ref{thm:univec}, we only need to show that the above SDP has value at least $(1-\delta-\beta)n$. To do this, we look at its dual program which is:
\begin{eqnarray}
	\textrm{Minimize }\sum_{i\in [n]}q_i \nonumber\\
	\textrm{s.t.} \qquad 
	\label{dualuv1} \sum_k \eta_k w_kw_k^T+G&-&\frac{1}{\beta}diag(G)+\sum_{i}q_ie_ie_i^T  \succeq I \\
   	\eta&\in&\mathbb{R} \label{dualuv3} \\
	G &\succeq &0 \label{dualuv6}\\
	q &\ge &0 \label{dualuv4} 
\end{eqnarray}
Here we have the dual variables $\eta_k$ and $q_i$ corresponding to the constraints (\ref{sdpuv1}),(\ref{sdpuv3}) respectively of the primal SDP, and we have a PSD matrix $G$ as the ``dual variable" for constraint (\ref{sdpuv2}). 

We show below that strong duality holds for this SDP, for which we use the following result. 
\begin{thm}[Theorem 4.7.1, \cite{GM12}]
\label{thm:duality}
If the primal program $(P)$ is feasible, has a finite optimum value $\eta$ and has an interior point $\tilde{x}$, then the dual program $(D)$ is also feasible and has the same finite optimum value $\eta$.
\end{thm}

\begin{lemma}
The SDP described above is feasible and has value equal to its dual program.
\end{lemma}
\begin{proof}
We apply Theorem~\ref{thm:duality}, with $P$ equal to the dual of the SDP. This would suffice as the dual $D$ of $P$ is our SDP. 

We claim that the following solution is a feasible interior point: $q_i=\frac{1+\epsilon}{\beta}$ for all $i\in [n]$ and $\epsilon>0$, $\eta_{k}=0$ for all $1\le k\le \ell$ and $G=I$.

It is clearly feasible as it satisfies the constraints \eqref{dualuv1}-\eqref{dualuv4}.
Moreover each of the constraints \eqref{dualuv1}-\eqref{dualuv4} is satisfied with strict inequality and hence this solution is an interior point.
%
As this point has objective value at most $\frac{1+\epsilon}{\beta}n$ and as the $q_i$'s are non-negative, $P$ has a finite optimum value. 
\end{proof}

As strong duality holds, to prove Theorem~\ref{thm:univec} it suffices to prove that for any feasible solution to the dual program, the objective value of the dual program is at least $(1-\delta-\beta)n$. The following lemmas will be useful.

\begin{lemma}
\label{col1matr}
Given an $m\times n$ matrix $M$ with columns $m_1,m_2,\dots,m_n$. If $\|m_i\|_2\le 1$ for all $i\in[n]$, then for any $\beta\in(0,1]$, there exists a subspace $W$ of $\mathbb{R}^n$ satisfying:
\begin{enumerate}[i)]
\item $dim(W)\ge (1-\beta)n $, and
\item $\forall y\in W$, $\|My\|_2^2 \le (\frac{1}{\beta})\|y\|_2^2$
\end{enumerate}
\end{lemma}
\begin{proof}
Let the singular value decomposition of $M$ be given by $M=\sum_{i=1}^n \sigma_i p_i q_i^T$, where $0\le \sigma_1\le\dots\le\sigma_n$ are the singular values of $M$ and $\{p_i : i\in [n]\}, \{q_i : i\in [n]\}$ are two sets of orthonormal vectors.
Then,
\[
\sum_{i=1}^n\sigma_i^2 =\textrm{Tr}[\sum_{i=1}^n \sigma_i^2 q_iq_i^T]=\textrm{Tr}[M^TM]=\sum_{i=1}^n \|m_i\|_2^2 \le n
\]
So at least $\lceil(1-\beta)n\rceil$ of the squared singular values $\sigma_i^2$s have value at most $(1/\beta)$, and thus $\sigma_1\le\dots\le\sigma_{\lceil(1-\beta)n\rceil} \le \sqrt{(1/\beta)}$.
Let $W=\textrm{span}\{q_1,\dots,q_{\lceil(1-\beta)n\rceil}\}$. For $y\in W$,
\begin{eqnarray}
\|My\|_2^2 & = & \|\sum_{i=1}^n \sigma_ip_iq_i^T y\|_2^2 
  =   \|\sum_{i=1}^{\lceil(1-\beta)n\rceil} \sigma_ip_iq_i^T y\|_2^2 \nonumber\\
 & = & \sum_{i=1}^{\lceil(1-\beta)n\rceil} \sigma_i^2 (q_i^Ty)^2  \qquad \textrm{(since $p_i$ are orthonormal)} \nonumber\\
 &\le & \frac{1}{\beta}\sum_{i=1}^{\lceil(1-\beta)n\rceil}(q_i^Ty)^2 \le \frac{1}{\beta}\|y\|_2^2 \qquad \textrm{(since $q_i$ are orthonormal)} \nonumber
\end{eqnarray}
\end{proof}

\begin{lemma}
\label{diagdom}
Given an $m\times n$ PSD matrix $G$ and $\beta\in(0,1]$, there exists a subspace $W\subseteq \mathbb{R}^n$ satisfying
\begin{enumerate}
\item $dim(W)\ge (1-\beta) n$
\item $\forall w\in W$, $w^TGw\le \frac{1}{\beta}w^Tdiag(G)w$
\end{enumerate}
\end{lemma}
\begin{proof}
Let $U$ be such that $G=U^TU$. Then, $w^TGw\le \frac{1}{\beta}w^Tdiag(G)w$ is equivalent to
\[ \|Uw\|_2^2\le\frac{1}{\beta}\|\sqrt{diag(G)}w\|_2^2.\]
Let $N \subseteq [n]$ be the set of coordinates $i$ with $G_{ii} > 0$. We claim that it suffices to focus on the coordinates in $N$.
Let us first observe that if $i \notin N$, i.e.~$G_{ii}=0$, the $i^{th}$ column of $U$ must be identically zero, and we have
\[ \|Ue_i\|_2^2= \frac{1}{\beta}\|\sqrt{diag(G)}e_i\|_2^2=0.\]
As the directions $e_i$ for $i\in N$ are orthogonal to the directions in $[n]\setminus N$, it 
suffices to show that there is a $(1-\beta)|N|$ dimensional subspace $W$ in $\textrm{span}\{e_i:i\in N\}$ such that $ \|Uw\|_2^2 - (1/\beta)  \|\sqrt{diag(G)}w\|_2^2 \leq 0$ for each $w \in W$.
The overall subspace we desire is simply $W \oplus \textrm{span}\{e_i: i \in [n]\setminus N\}$ which has dimension $(1-\beta)|N| + (n-|N|) \geq (1-\beta)n$.

So, let us assume that $N=[n]$ (or equivalently restrict $G$ and $U$ to columns in $N$), which gives us that $G_{ii} > 0$ for all $i\in N$ and hence that $diag(G)$ is invertible.

Let $\tilde{U}=U diag(G)^{-1/2} $. The $\ell_2$-norm of each column in $\tilde{U}$ is $1$, and by Lemma \ref{col1matr}, there is a subspace $\tilde{W}$ of dimension at least $(1-\beta)|N|$ such that $\|\tilde{U}\tilde{y}\|_2^2 \leq  \frac{1}{\beta} \|\tilde{y}\|_2^2$ for each $\tilde{y}\in \tilde{W}$.
Setting $y=diag(G)^{-1/2}\tilde{y}$ gives
 \[ \|U y \|_2^2 =  \|\tilde{U}\tilde{y}\|_2^2 \leq \frac{1}{\beta} \|\tilde{y}\|_2^2 = \frac{1}{\beta} \|\sqrt{diag(G)}y\|_2^2, \]  
and thus $W = \{diag(G)^{-1/2}\tilde{y}: \tilde{y}\in \tilde{W}\}$ gives the desired subspace as $\textrm{dim}(W)= \textrm{dim}(W')$. 

\end{proof}

We now show the main result of this section which will imply Theorem~\ref{thm:univec}. 
\begin{thm}
The optimum value of the dual program is at least $(1-\delta-\beta)n$.
\end{thm}
\begin{proof}
Consider some feasible solution specified by $\eta, G,q$.
Let $C=\sum_k \eta_k w_kw_k^T$ and let $Q$ be a diagonal matrix with $i^{th}$ diagonal entry equal to $q_i$ i.e. $Q=\sum_i q_ie_ie_i^T$. 
As $C$ is the sum of $\ell=\delta n$ rank one matrices, $rank(C) \leq \delta n$. 
Let $C_{\perp}$ denote the subspace orthogonal to $C$ and let $W'=W\cap C_{\perp}$ where $W$ is the subspace obtained by Lemma~\ref{diagdom} applied to $G$. Then, 
\[ dim(W')\ge dim(W)-dim(C)\ge (1-\beta) n-\delta n=(1-\beta-\delta)n\]

Let $p_1,\dots,p_d$ with $d=dim(W')$ form an orthonormal basis of $W'$ and let $M=\sum_{i=1}^d p_ip_i^T$ denote the projection matrix onto the span of $W'$. Taking the trace inner product of \eqref{dualuv1} with $M$, we get
\begin{eqnarray*}
I\bullet M & \le&  (C+G-\frac{1}{\beta}diag(G)+Q)\bullet M \\
&\le &  Q\bullet M
\end{eqnarray*}
where the inequality uses $C\bullet M=0$ (as $W' \subseteq C_{\perp}$) and $(G-\frac{1}{\beta}diag(G))\bullet M \le 0$ (as $W' \subseteq W$ and using Lemma~\ref{diagdom}). Taking trace on both sides,
\[ tr(Q\bullet M)\ge tr(I\bullet M)=rank(I\bullet M)=dim(W')\ge (1-\beta-\delta)n .\]
Notice now that 
\[tr(Q\bullet M)=\sum_i{q_i m_{ii}}=\sum_i q_i e_i^T Me_i\le \sum_i q_i \|e_i\|_2^2=\sum_i q_i\]
where the inequality follows since $M$ is a projection matrix. This completes the proof.
\end{proof}

\section{Tusnady's Problem}
\label{sec:tusnady}




Given a set $P$ of $n$ points in $\mathbb{R}^d$, let us first observe that only $n^{2d}$ distinct axis-parallel boxes matter. This is because any
axis-parallel box can be shrunk to have a point of $P$ on all of its $(n-1)$-dimensional facets while not changing the set of points contained in that box. Because there are $2d$ such facets and there is a unique axis-parallel box having a fixed set of $2d$ points on its facets, there can be at most $n^{2d}$ distinct axis-parallel boxes.


The proof of Theorem~\ref{thm:tusnady} follows 
directly by applying the framework in Theorem~\ref{thm:vecgenunified} to the previous proof \cite{Mat09}.
In particular,
the previous proof uses a construction of \textit{canonical boxes} stated below (for completeness we provide the full construction in Appendix C), and applies the partial coloring method while requiring that the canonical boxes incur zero discrepancy. The errors then add up over the $O(\log n)$ phases. 
For our better bound, we also choose the linear constraints $Z(t)$ as the incidence vector of these canonical boxes, but roughly speaking, our errors only add up in an $\ell_2$ manner, instead of in an $\ell_1$ manner.

\begin{lemma}
\label{lem:canonical}
\emph{(Canonical Boxes in $\mathbb{R}^d$)}
For any set $P$ of $n$ points in $\mathbb{R}^d$, there exists a (constructible) collection $\mathcal{B}(P)$ of $n/8$ axis-parallel boxes such that any axis-parallel box $R$ in $\mathbb{R}^d$ can be expressed as a disjoint union of some boxes from $\mathcal{B}(P)$ and a small region $R'\subseteq R$, where $R'$ contains at most $O_d(\log^{2d-2} n)$ points of $P$.
\end{lemma}


\begin{proof}(of Theorem~\ref{thm:tusnady})
We use Theorem~\ref{thm:vecgenunified}. At time step $t$, choose $A(t)$ to be the set of all alive elements at time $t$ i.e. $A(t)=N(t)$. 

The constraints in $Z(t)$ are chosen as follows.
At time $t$, a particular subset $P'$ is chosen and the canonical boxes in $\mathcal{B}(P')$ are constructed according to Lemma~\ref{lem:canonical}. Let $Box(t)$ denote the set of canonical boxes we are going to construct at time $t$. Initially we set $Box(0)=\mathcal{B}(P)$. For a time $t$, let $k_t$ be an integer such that $\frac{n}{2^{k_t+1}}<|N(t)|\le \frac{n}{2^{k_t}}$ and let $t'\le t$ be the first time when the number of alive elements $|N(t')|$ was at most $\frac{n}{2^{k_t}}$. Then, $Box(t)=\mathcal{B}(N(t'))$. It easily follows that
\[|Box(t)|=|\mathcal{B}(N(t'))|\le \frac{|N(t')|}{8}\le \frac{n}{8.2^{k_t}}=\frac{n}{4.2^{k_t+1}}\le \frac{|N(t)|}{4}.\]

We now put the indicator vectors of canonical boxes in $Box(t)$ as the constraints in $Z(t)$. It follows that
\[|Z(t)| =|Box(t)|\le |N(t)|/4=|A(t)|/4 ,\]
giving $\delta\le 1/4$.

At any time $t$, the number of points not protected in an axis-parallel box $R$, given by the set $R'$, is at most $O_d(\log^{2d-2}n)$. Because $Z(t)$ only changes $\log n$ times during the algorithm, there are $O_d(\log^{2d-1}n)$ corrupted points in $R$. Using Theorem~\ref{thm:vecgenunified} now with $\lambda=O(\sqrt{d\log n})$ gives that the discrepancy of $R$ is $O_d(\log^{d} n)$ with probability at least $1-1/n^{3d}$. The result follows now by taking a union bound over all the distinct $n^{2d}$ axis-parallel boxes.
\end{proof}

\section{Steinitz Problem}
\label{sec:steinitz}


In this section we prove Theorem~\ref{thm:stff} and Theorem~\ref{thm:st22}. This immediately gives the same bounds for Steinitz problem in $\ell_\infty$ and $\ell_2$ norms by the (constructive) reduction from Steinitz to the signed series problem given in \cite{Har14}.

We start with a simple observation.

\begin{obs}
\label{obstrivial}
We can assume $d > \log n$. Otherwise the (algorithmic) bound of $d$ \cite{GS90} already gives us a bound of $O(\sqrt{d\log n})$. We can also assume $d<n^2$, otherwise a bound of $\sqrt{d}$ is trivial.
\end{obs}


\begin{proof}(of Theorem~\ref{thm:stff})
We will apply Theorem~\ref{thm:vecgenunified}.
At time $t$, take $A(t)$ to be the first $2d$ alive elements i.e. we include in $A(t)$ the smallest $2d$ indices in the set $N(t)$ of alive elements. If there are fewer than $2d$ elements alive, then we take $A(t)$ to contain all the alive elements. 

In $Z(t)$ we will include all the $d$ rows restricted to $A(t)$ if $A(t)$ has at least $2d$ elements in it i.e. we include in $Z(t)$ the vector $w_j\in \mathbb{R}^{|A(t)|}$ for $1\le j\le d$ with $w_j(i)$ equal to the $j^{th}$ coordinate of the $i^{th}$ element in $A(t)$.
If $A(t)$ has less than $2d$ elements, we take $Z(t)$ to be the null collection. Thus, 
\[ |Z(t)|\le  A(t)/2  \quad \text{ for all }t\]
and $\delta=|Z(t)|/|A(t)|\le 1/2$.

Fix a row $j$ and a prefix $k\in[n]$. Let $v_k$ be first included in $A(t)$ at time $t_k$. 
Then for $t<t_k$, $A(t)\subseteq \{1,2,\dots,k-1\}$ and because of our choice of $w_j$, every element in the set $S_k=\{1,2,\dots,k\}$ is protected at time $t$ for $(j,S_k)$.
For $t\ge t_k$, only the elements in $A(t_k)$ can ever become corrupt for $(j,S_k)$ because trivially all other elements are either not included in $S_k$ or have been frozen by now and thus will not be included in the set $A(t)$ at any time $t\ge t_k$. Thus, the set of corrupted elements for $(j,S_k)$ is a subset of $A(t_k)$, giving
\[|\mathcal{C}_{j,S_k}|\le |A(t_k)| \le 2d .\]
As each entry $a_{ji}$ is at most $1$ in absolute value, we get $\sum_{i\in \mathcal{C}_{j,S_k}}a_{ji}^2\le 2d$. Using Theorem~\ref{thm:vecgenunified} now with $\lambda=O(\sqrt{\log n})$, we get that discrepancy of row $j$ and prefix $k$ is $O(\sqrt{d\log n})$ with probability at least $1-1/n^4$. Taking a union bound over the $n$ prefixes and at most $n^2$ rows (by Observation~\ref{obstrivial}) finishes the proof.
\end{proof}

\begin{proof} (of Theorem~\ref{thm:st22})
This follows similar to the previous proof.
We choose $A(t)$ and $Z(t)$ exactly as before. Thus, for every row $j$ and prefix $k$, $\mathcal{C}_{j,S_k}$ is a subset of $A(t_k)$.
Then we have,
\[ \sum_j \sum_{i\in \mathcal{C}_{j,S_k}}a_{ji}^2 \le \sum_j \sum_{i\in A(t_k)}a_{ji}^2=\sum_{i\in A(t_k)}\sum_j a_{ji}^2 \le |A(t_k)|\le 2d .\]
The second last inequality follows as for all $i$, $\sum_j a_{ji}^2=\|v_i\|_2^2\le 1$. Using Theorem~\ref{thm:l2discrepancy} now with $\lambda=O(\sqrt{\log n})$ gives that the $\ell_2$ discrepancy of prefix $k$ is $O(\sqrt{d\log n})$ with high probability. Taking a union bound over all prefixes finishes the proof.
\end{proof}

\section*{Acknowledgments}
We would like to thank Yin-Tat Lee for sharing with us a simpler proof of the martingale concentration in \cite{BDG16}; and Daniel Dadush, Aleksandar Nikolov and Mohit Singh for insightful discussions related to discrepancy over the past year. 
We would also like to thank the American Institute of Mathematics, for hosting a workshop organized by Aleksandar Nikolov and Kunal Talwar, where this work was started.

\bibliographystyle{alpha}
\bibliography{refr}

\section*{Appendix A: Proof of the Framework}

\label{sec:appa}

Let $\epsilon>0$ be a constant such that $\delta\le 1-\epsilon$ always.
Let $\gamma=\frac{1}{n^{10}m^4\log(mn)}$ and $T=\frac{6}{\epsilon\gamma^2}n\log n$. 
The algorithm is as follows.\\

\noindent 
{\bf Algorithm:}
\begin{enumerate}
\item Initialize $x_0(i) =0$ for all $i\in [n]$. 
\item  For each time step $t=1,2,\ldots,T$ while $N(t)\neq\phi$ repeat the following:
\begin{enumerate}
\item Take as input $A(t)$ and $Z(t)$
\item \label{apx:step1}
Use Theorem~\ref{thm:univec} with the set of elements $A(t)$, $w_k$'s as the $\delta|A(t)|$ linear constraints in $Z(t)$, and $\beta=(1-\delta)/2$ to get a universal vector coloring $u_i^t$ for $i\in A(t)$.
\item 
\label{apx:round}
Let $r_t \in \R^n$ be a random $\pm 1$ vector, obtained by setting each coordinate $r_t(i)$ independently to $-1$ or $1$ with probability $1/2$.

For each $i \in A(t)$, update $x_t(i)=x_{t-1}(i)+\gamma\langle r_t, u_i^t\rangle$. 
For  each $i \not\in A(t)$, set $x_t(i)=x_{t-1}(i)$.

\end{enumerate}
\item 
\label{stp3}
Generate the final coloring as follows.
For the frozen elements $i\notin N(T+1)$, set $x_T(i)=1$ if $x_T(i)\geq 1-1/n$ and $x_T(i)=-1$ otherwise.
For the alive elements $i \in N(T+1)$
set them arbitrarily to $\pm 1$. 
\end{enumerate}

\subsection*{Analysis}
For convenience, we will set $u_i^t=0$ for all $t$ and $i\not\in A(t)$. Notice that $|\gamma\langle r_t,u_i^t\rangle|=o(1/n)$ and thus no $x_t(i)$ will ever exceed $1$ in absolute value. Theorem~\ref{thm:univec} directly gives the following lemma.
\begin{lemma}
\label{lem:sdpvalue}
At each time $t$, $\sum_{i}\|u_i^t\|_2^2 \ge (1-\delta)|A(t)|/2$.
\end{lemma}

We need to show now that by time $T$, discrepancy is small and that all elements are colored by time $T$ with high probability. The following simple lemma will be needed several times.
\begin{lemma}
\label{lem:randvector}
For any vector $v\in \mathbb{R}^n$ and a random vector $r$ distributed uniformly in $\{-1, 1\}^n$, $\E[\langle r,v \rangle]=0$ and $\E[\langle r,v\rangle^2]=\| v\|_2^2$.
\end{lemma}
\begin{proof}
$\E[\langle r,v\rangle]=\sum_i\E[r_iv_i]=0$, giving the first part of the lemma. For the second part, 
\begin{equation*}
\E[\langle r,v\rangle^2]  = \sum_{i,j}\E[r_ir_jv_iv_j]=\sum_i v_i^2.
\end{equation*}
We use that for $i\neq j$, $\E[r_ir_j]=0$.
\end{proof}
We first show that by time $T$, each element is frozen with high probability.

\begin{lemma}
\label{thm:terminate}
After time $T=\frac{6}{\epsilon\gamma^2}n\log n$, there are no alive variables left with probability at least $1-O(n^{-2})$.
\end{lemma}
\begin{proof}
Given the coloring $x_t$ at time $t$, define $G_t = \sum_{i \in [n]} (1-x_t(i)^2)$. Clearly $G_t\leq n$ for all $t$. 
As $x_{t}(i) =   x_{t-1}(i) + \gamma \langle r_t,u_i^t \rangle $, using Lemma~\ref{lem:randvector}, we have that $\E_{t-1}[x_t(i)^2] = x_{t-1}(i)^2 + \gamma^2 \|u_i^t\|_2^2$.
It follows
\begin{eqnarray*} 
\E_{t-1}[G_t]  & = &   \E_{t-1}\left[\sum_{i \in [n]} (1-x_{t}(i)^2)\right] \\
                        &= & \sum_{i \in [n]} \left(1-x_{t-1}(i)^2 \right) - \gamma^2 \sum_{i \in [n]} \|u_i^t\|_2^2   \\
                  & \leq &  \sum_{i \in[n]} \left(1-x_{t-1}(i)^2 \right) - (1-\delta)\gamma^2  |A(t)|/2  \qquad\textrm{(using Lemma~\ref{lem:sdpvalue})}\\
                   &\leq & G_{t-1}-(1-\delta)\gamma^2/2  \qquad\textrm{($|A(t)|\ge 1$ at every time)}\\ 
                  & \leq &  (1 - \frac{(1-\delta)\gamma^2}{2n}) G_{t-1}  
\end{eqnarray*}
Thus by induction, 
\[\E[G_{T+1}] \leq (1-\frac{(1-\delta)\gamma^2}{2n})^T G_1 \leq n. e^{-(1-\delta)\gamma^2 T/2n } \le 1/n^3.\] 
Thus by Markov's inequality, $\Pr[G_{T+1} \geq 1/n] \leq 1/n^2$. However, $G_{T+1} \leq 1/n$ implies that all variables are frozen as each alive variable contributes at least $1-(1-1/n)^2 > 1/n$ to $G_{T+1}$.
\end{proof}

To bound discrepancy, we will use a concentration inequality which is a variant of Freedman's inequality for martingales\cite{Freedman}. The following lemma will be useful.

\begin{lemma}
\label{lem:taylor}
Let $X$ be a random variable such that $|X|\leq 1$. Then for any $\theta >0$, 
\[\E[e^{\theta X}] \leq  e^{\theta \E [X] +(e^\theta-\theta-1)\E [X^2]}\]
\end{lemma}
\begin{proof}
Let $f(x) = \frac{e^{\theta x} - \theta x - 1}{x^2}$ where we set $f(0)=\theta^2/2$. Then $f(x)$ is increasing for all $x$.
This implies $e^{\theta x} \leq f(1)x^2 + 1 + \theta x$ for any $x\le 1$.
Taking expectation, this becomes 
\begin{eqnarray*}
\E[e^{\theta X}] &\leq & 1+\E[\theta X]+f(1)\E [X^2] \\
& = & 1+\E[\theta X]+(e^\theta-\theta-1)\E [X^2] \\
& \le &  e^{\theta \E [X] +(e^\theta-\theta-1)\E [X^2]}
\end{eqnarray*}
where the last inequality uses the fact that $1+x\le e^x$.
\end{proof}

We will use the following concentration inequality to bound the discrepancy. This is a slight modification of Freedman's inequality due to Yin-Tat Lee and we show its proof below. 

\begin{thm}
\label{thm:freedman}
Let $Y_1,\ldots,Y_n$ be a sequence of random variables with $Y_0=0$ such that for all $t$,
\begin{enumerate}[i)]
\item $|Y_t - Y_{t-1}|\leq 1$, and
\item $\E_{t-1}[Y_t - Y_{t-1}] \le -\alpha\E_{t-1}[(Y_t - Y_{t-1})^2]$
\end{enumerate}
for some $\alpha\in(0,0.5)$.
Then for all $\lambda \geq 0$ , we have 
\[ \Pr[Y_n  \geq \lambda ]\le \exp\left(-\alpha\lambda\right). \] 
\end{thm}
\begin{proof}
Let $X_t=Y_t-Y_{t-1}$ be the difference sequence and $\theta>0$ be a real number to be determined later. Then using Markov's inequality,
\begin{equation}
\label{markov}
 \Pr[Y_n\ge \lambda] \le \Pr[e^{\theta Y_n}\ge e^{\theta \lambda}] \le \frac{\E[e^{\theta Y_n}]}{e^{\theta\lambda}} 
 \end{equation}
To bound $\E[e^{\theta Y_n}]$ we observe the following:
\begin{eqnarray*}
\E_{t-1}[e^{\theta Y_t}] &= & e^{\theta Y_{t-1}}\E_{t-1}[e^{\theta X_t}] \\
& \le &  e^{\theta Y_{t-1}} \exp\left(\theta \E_{t-1}[X_t] +(e^\theta-\theta-1)\E_{t-1}[X_t^2]\right) \qquad \textrm{(using Lemma~\ref{lem:taylor})} \\
& \le & e^{\theta Y_{t-1}}\exp\left((e^\theta-\theta-1-\alpha\theta)\E_{t-1}[X_t^2]\right) \qquad \textrm{(using }\E_{t-1}[X_t] \le -\alpha\E_{t-1}[X_{t-1}^2] )\\
\end{eqnarray*}
Set $\theta$ to be the solution of 
\[ \alpha =\frac{e^\theta-\theta-1}{\theta}\]
and notice that for $\alpha\in(0,0.5)$, it must hold that $\alpha\le \theta$. Then
we get $\E_{t-1}[e^{\theta Y_t}] \le e^{\theta Y_{t-1}} $. And thus by induction, 
\[ \E[e^{\theta Y_n}] \le \E[e^{\theta Y_0}]=1 .\]
Putting this in (\ref{markov}), we get
\[\Pr[Y_n\ge \lambda] \le e^{-\theta\lambda}\le e^{-\alpha\lambda}. \]
\end{proof}

We can in fact strengthen the above Theorem to get the same probability bound that the sequence $\{Y_t\}$ never exceeds the value $\lambda$.
\begin{cor}
\label{cor:freedmanbetter}
Let $Y_1,\ldots,Y_n$ be a sequence of random variables with $Y_0=0$ such that for all $t$,
\begin{enumerate}[i)]
\item $|Y_t - Y_{t-1}| \leq 1$, and
\item $\E_{t-1}[Y_t - Y_{t-1}] \le -\alpha\E_{t-1}[(Y_t - Y_{t-1})^2]$
\end{enumerate}
for some $\alpha\in(0,0.5)$.
Then for all $\lambda \geq 0$ , we have 
\[ \Pr[\exists t: Y_t  \geq \lambda ]\le \exp\left(-\alpha\lambda\right). \] 
\end{cor}
\begin{proof}
Define a new sequence $\tilde{Y}$ defined by 
\begin{align*}
\tilde{Y}_t =
\left\{
	\begin{array}{ll}
		Y_t  & \mbox{if } Y(t')< \lambda \text{ for all }t'\le t \\
		\lambda & \mbox{otherwise } 
	\end{array}
\right.
\end{align*}
i.e. $\tilde{Y}_t$ equals $Y_t$ as long as $Y_t$ is less than $\lambda$. If and when $Y_t$ equals (or exceeds) $\lambda$ for the first time, we stick $\tilde{Y}$ to $\lambda$ and take the further increments of the sequence to be zero. Now we just apply Theorem~\ref{thm:freedman} on $\tilde{Y_t}$ to get
\[ Pr[\tilde{Y}_n  \geq \lambda ]\le \exp\left(-\alpha\lambda\right).\] 
But $Pr[\tilde{Y}_n  \geq \lambda ]$ is equivalent to $Pr[\exists t\le n: {Y}_t  \geq \lambda ]$, finishing the proof.
\end{proof}

\subsection*{Proof of Theorem~\ref{thm:vecgenunified}: Bounding the $\ell_\infty$ discrepancy}
\label{sec:appb}


It should be pointed out why we only get a sub-gaussian bound for $\lambda\le (\sum_i a_{ji}^2)^{1/2}$ and need to add the $+\lambda$ term. For instance if the number of corrupted elements for any $(j,S)$ pair is much smaller than $\lambda$, we might not get a concentration as strong as a sub-gaussian. An example of such a situation can be if for a set $S=\{1,2,\dots,p\}$, at every time $t$, $Z(t)$ contains exactly one of $w_1=\{1,2,\dots,p-1\}$ and $w_2=\{2,3,\dots,p\}$. Now suppose it happens that $w_1\in Z(t)$ until the fractional color of $p$ reaches almost $+1$ and meanwhile fractional color of point $1$ has reached almost $-1$. Then we include $w_2$ in $Z(t)$ (and do not include $w_1$). Now the algorithm sets the fractional color of point $1$ close to $+1$ while making the fractional color of point $p$ close to $-1$. Repeating this many times, we see that this set $S$ has only two corrupted points, but its discrepancy can become quite larger than $\sqrt{2}$. 

We now come to the proof of Theorem~\ref{thm:vecgenunified}. 
\begin{proof}(of Theorem~\ref{thm:vecgenunified})

The final rounding in step~\ref{stp3} of the algorithm can only affect the discrepancy of every set by at most $\sum_i |a_{ji}|/n\le 1$ as every entry $a_{ij}$ is at most $1$ in absolute value. Thus we can safely ignore the effect on discrepancy due to this. 

Fix a row $j$ and a set $S$. Let $C_{j,S}(t)$ denote the set of elements that were not protected for $(j, S)$ at any time before and including $t$. Notice that $C_{j,S}(T)=C_{j,S}$.

Also, let $C_t$ denote the set of elements that were not protected for $(j,S)$ at time $t$. It holds that $C_{j,S}(t)=C_{j,S}(t-1)\cup C_t$. 
For an element $i$, let $t_i$ be the last time $i$ was not included in $C_{j,S}(t)$ i.e. $t_i+1$ is the first time when element $i$ is not protected for $(j,S)$ and becomes corrupt.

Let $disc(t)=\sum_{i\in S} a_{ji}x_t(i)$ be the signed discrepancy for set $S$ and row $j$ at time $t$. 
Also
define the energy of this set at time $t$ as $E(t)=\sum_{i\in  C_{j,S}(t)} a_{ji}^2[x_t(i)^2-x_{t_i}(i)^2]$.
Then these two quantities change with time as follows:
\[ disc(t)-disc(t-1)=\sum_{i\in S} a_{ji}\Delta x_t(i)=\gamma\langle r_t, \sum_{i\in S} a_{ji}u_i^t\rangle=\gamma\langle r_t, \sum_{i\in C_t} a_{ji}u_i^t\rangle .\]
The last equality follows since the algorithm is required to give zero discrepancy to all sets in $Z(t)$ and by the definition of $C_t$. The energy changes with time as
\begin{eqnarray*}
E(t)-E(t-1) &= & \sum_{i\in  C_{j,S}(t-1)} a_{ji}^2[x_t(i)^2-x_{t-1}(i)^2]+\sum_{i\in C_{j,S}(t)\setminus C_{j,S}(t-1)} a_{ji}^2[x_t(i)^2-x_{t_i}(i)^2]\\
& = &  \sum_{i\in  C_{j,S}(t)} a_{ji}^2[x_t(i)^2-x_{t-1}(i)^2] \\
&=& \gamma^2\sum_{i \in C_{j,S}(t)} a_{ji}^2\langle r_t,u_i^t\rangle^2+2\gamma\langle r_t,\sum_{i\in C_{j,S}(t)}a_{ji}^2x_{t-1}(i)u_i^t\rangle
\end{eqnarray*}
The second equality uses the fact that all $i\in C_{j,S}(t)\setminus C_{j,S}(t-1)$ got corrupted for the first time at time $t$ and hence for them $t_i=t-1$.
Define the sequence $\{Y_t\}$ as
\[Y_t=disc(t)-\eta E(t)\]
where the exact value of $\eta$ will be determined later. Then,
\[Y_t-Y_{t-1}=\gamma\langle r_t, \sum_{i\in C_t} a_{ji}u_i^t\rangle -\eta \gamma^2\sum_{i \in C_{j,S}(t)} a_{ji}^2\langle r_t,u_i^t\rangle^2-2\eta\gamma\langle r_t,\sum_{i\in C_{j,S}(t)}a_{ji}^2x_{t-1}(i)u_i^t\rangle.\]
Using Lemma~\ref{lem:randvector}, this gives
\[ \E_{t-1}[Y_t-Y_{t-1}]=-\eta\gamma^2\sum_{i\in C_{j,S}(t)}a_{ji}^2\|u_i^t\|_2^2\]
and
\begin{eqnarray*}
&&\E_{t-1}[(Y_t-Y_{t-1})^2]=\gamma^2\| \sum_{i\in C_t}a_{ji}u_i^t-2\eta\sum_{i\in C_{j,S}(t)}a_{ji}^2x_{t-1}(i)u_i^t\|_2^2+O(\gamma^3n^4) \qquad\textrm{(using Lemma~\ref{lem:randvector})}\\
& \le & (1/\beta)\gamma^2 \left(\sum_{i\in C_{j,S}(t)\setminus C_t}4\eta^2a_{ji}^4x_{t-1}(i)^2\|u_i^t\|_2^2+\sum_{i\in C_t}(a_{ji}-2\eta a_{ji}^2x_{t-1}(i))^2 \|u_i^t\|_2^2\right)+O(\gamma^3n^4)  \\
& \le & (4/\epsilon)\gamma^2 \sum_{i\in C_{j,S}(t)} a_{ji}^2 \|u_i^t\|_2^2+O(\gamma^3n^4) \qquad\textrm{(using $\beta\ge \epsilon/2$ and $|a_{ji}|\le 1$)}
\end{eqnarray*}
where the first inequality uses property (ii) of Universal Vector Colorings and the last inequality holds for $\eta\le 1/5$. We are going to use Theorem~\ref{thm:freedman} with the first term above and can safely ignore the $O(\gamma^3n^4)$ term. This is because the $O(\gamma^3n^4)$ term can only contribute $O(T\gamma^3n^4)=o(1)$ to the total variance and hence to $Y_T-Y_{0}$. 

Using Theorem~\ref{thm:freedman} now with $\alpha=\epsilon\eta/4$, and noting that 
\[ Y_T=disc(T)-\eta E(T) \ge disc(T)-\eta\sum_{i\in C_{j,S}} a_{ji}^2\]
we get
\begin{eqnarray*}
\Pr\left[disc(T) \ge c\lambda\left((\sum_{i\in C_{j,S}}a_{ji}^2)^{1/2}+\lambda\right)\right] & \le &  \Pr\left[Y_T \ge c\lambda\left((\sum_{i\in C_{j,S}}a_{ji}^2)^{1/2}+\lambda\right)-\eta \sum_{i\in C_{j,S}} a_{ji}^2\right] \\
& \le & \exp\left( \frac{-\eta \epsilon c\lambda\left((\sum_{i\in C_{j,S}}a_{ji}^2)^{1/2}+\lambda\right)+\eta^2\epsilon \sum_{i\in C_{j,S}} a_{ji}^2}{4}  \right )
\end{eqnarray*}
For $\lambda\le (\sum_{i\in C_{j,S}(t)} a_{ji}^2)^{1/2}/5c$, choosing $\eta =\frac{c\lambda(\lambda+(\sum_{i\in C_{j,S}} a_{ji}^2)^{1/2})}{2\sum_{i\in C_{j,S}} a_{ji}^2}$, we get
\begin{eqnarray*}
\Pr\left[disc(T) \ge c\lambda\left((\sum_{i\in C_{j,S}}a_{ji}^2)^{1/2}+\lambda\right)\right]   \le  \exp\left( -\lambda^2/2  \right ) 
\end{eqnarray*}
for $c\ge \sqrt{8/\epsilon}$. 
If $\lambda\ge (\sum_{i\in C_{j,S}(t)} a_{ji}^2)^{1/2}/5c$, we choose $\eta=1/5$ to get
\[ \Pr\left[disc(T) \ge c\lambda\left((\sum_{i\in C_{j,S}}a_{ji}^2)^{1/2}+\lambda\right)\right] \le  \exp\left( -\lambda^2/2  \right ) \] 
for $c\ge 10/\epsilon$.
Noting now that $disc(T)$ has a symmetric probability distribution around zero finishes the proof for Theorem~\ref{thm:vecgenunified}. 
\end{proof}

We mention now how to handle the generalisation mentioned in Remark~\ref{rem:relax}. In that case, the increment in discrepancy changes as
\[ disc(t)-disc(t-1) = \gamma \langle r_t, \sum_{i\in C_t} b_{ji} u_{i}^t  \rangle  \]

where $b_{ji}=a_{ji} - v_{S',j,t} (i)$ and $|b_{ji}| =O(|a_{ji}|)$. Energy is still defined in the same way as $E(t)=\sum_{i\in  C_{j,S}(t)} a_{ji}^2[x_t(i)^2-x_{t_i}(i)^2]$. The only place where the proof changes is in the calculation of $\E_{t-1}[(Y_t-Y_{t-1})^2]$; this is still calculated as before but this quantity can increase by a constant factor. But this only affects the value of $\alpha$ by a constant factor and the rest of the proof goes through in the same manner.

\subsection*{Proof of Theorem~\ref{thm:l2discrepancy}: Bounding the $\ell_2$ discrepancy}


For slightly simpler computations, we prove the following Theorem and notice that it suffices to prove Theorem~\ref{thm:l2discrepancy} by replacing $\lambda$ in the following Theorem by $(\sqrt{32/\epsilon})\lambda$. This only affects the constant $c$.

\begin{thm}
\label{thm:l2disc}
There is a constant $c>0$ such that given a matrix $B$ with each column of $\ell_2$ norm at most $1$ (i.e. $\|B\|_2\le 1$), then for any set $S\subseteq [n]$ and $\lambda\ge 0$, the coloring $\chi\in\{-1,1\}^n$ returned by the algorithm satisfies
\[ \Pr\left[\left(\sum_j disc(\chi,j,S)^2\right)^{1/2}\ge c \lambda\left((\sum_j \sum_{i\in C_{j,S}} a_{ji}^2)^{1/2}+\lambda\right)\right] \le 3e^{-\epsilon\lambda^2/64} .\]
\end{thm}

The increase in squared $\ell_2$ discrepancy due to the rounding in step~\ref{stp3} of the algorithm is at most $\sum_j (\sum_i a_{ji}/n)^2$ which is at most $1$ as $\sum_j a_{ji}^2\le 1$ for all $i$. We ignore the effect due to this rounding in the rest of the analysis.

\begin{obs}
\label{obs:trivial}
We can assume that $\lambda^2\ge 64/\epsilon$, as otherwise the right hand side in the equation above is at least one and Theorem~\ref{thm:l2disc} is trivially true.
\end{obs}



Fix a set $S\subseteq[n]$. Let 
$ disc_{t}(j,S)=\sum_{i\in S} a_{ji}x_t(i)$
denote the signed discrepancy of set $S$ and row $j$ at time $t$. Define
\[ D(t)^2=\sum_{j=1}^m disc_{t}(j,S)^2\]
to be the squared $\ell_2$ discrepancy of $S$ at time $t$. 
Let $H=c^2\lambda^2 M $, where we use $M$ as a shorthand for $\left((\sum_j \sum_{i\in C_{j,S}} a_{ji}^2)^{1/2}+\lambda\right)^{2}$. A variable will be called \textit{huge} if its value is at least $H$ and \textit{small} otherwise.
Define 
\[\Delta D(t)^2 = D(t)^2-D(t-1)^2 \]
to be the change in squared $\ell_2$ discrepancy at time $t$. Let $C_j^t$ be the corrupted (not protected) elements for the pair $(j,S)$ at time $t$. Then, 
\begin{align}
\Delta D(t)^2 & = D(t)^2-D(t-1)^2 
 = \sum_{j=1}^m disc_{t}(j,S)^2-disc_{t-1}(j,S)^2 \nonumber\\
& = \sum_{j=1}^m (\sum_{i\in S}a_{ji}x_t(i))^2-(\sum_{i\in S}a_{ji}x_{t-1}(i))^2 \nonumber\\
& = \sum_{j=1}^m \sum_{i,l\in S}a_{ji}a_{jl}[x_t(i)x_t(l)-x_{t-1}(i)x_{t-1}(l)] \nonumber\\
& = \sum_{j=1}^m \sum_{i,l\in S}a_{ji}a_{jl} [\gamma^2\langle r_t,u_i^t\rangle\langle r_t,u_l^t\rangle+\gamma\langle r_t,x_{t-1}(l)u_i^t+x_{t-1}(i)u_l^t\rangle]\nonumber\\
& = \gamma^2\sum_{j=1}^m \langle r_t, \sum_{i\in S}a_{ji}u_i^t\rangle^2 +2\gamma\langle r_t, \sum_{j=1}^m disc_{t-1}(j,S)\sum_{i\in S}a_{ji}u_i^t\rangle \nonumber\\
& = \underbrace{\gamma^2\sum_{j=1}^m \langle r_t, \sum_{i\in C_j^t}a_{ji}u_i^t\rangle^2}_\text{$\Delta_t Q$} +\underbrace{2\gamma\langle r_t, \sum_{j=1}^m disc_{t-1}(j,S)\sum_{i\in C_j^t}a_{ji}u_i^t\rangle}_\text{$\Delta_t L$} \label{fmarequiem}
\end{align}

The last equality follows by definition of corrupt elements $C_j^t$ and property (i) of Universal Vector Colorings.

Let the first quadratic term in (\ref{fmarequiem}) be $\Delta_t Q$ and the second linear term be $\Delta_t L$. Let $Q_t=\sum_{t'\le t}\Delta_{t'}Q$ and $L_t=\sum_{t'\le t}\Delta_{t'}L$. 

Let $F$ denote the event $\{\forall t: Q_t\le (8/\epsilon) M\}$ and let $G_t$ denote the event $\{\forall t'<t: D(t')^2<H\}$ i.e. the squared $\ell_2$ discrepancy is always small till time $t-1$. We claim the following concentrations for the terms $L_t$ and $Q_t$. We defer their proof for the moment and first show how they imply the theorem.
\begin{claim}
\label{claim:sketchlinear}
$\Pr[\exists t: L_t \ge (c^2-1)\lambda^2N \cap G_t]\le e^{-\epsilon\lambda^2/64}$.
\end{claim}

\begin{claim}
\label{claim:sketchquad}
$ \Pr[\overline{F}]=\Pr [\exists t: Q_t>(8/\epsilon) M]\le e^{-\epsilon\lambda^2/64} $.
\end{claim}

\begin{proof}(of Theorem~\ref{thm:l2disc})
We can bound the probability in theorem statement as
\begin{align}
\Pr[\exists t: D(t)^2\ge H] & \le \Pr[\exists t: L_t \ge H -(8/\epsilon)M \cap F ] +\Pr[\overline{F} ] \nonumber\\
& \le   \Pr[\exists t: L_t \ge (c^2-1)\lambda^2M \cap F ] +Pr[\overline{F} ] \label{end3}
\end{align}
The first inequality uses that $D(t)^2=L_t+Q_t$, and hence $D(t)^2\ge H$ and $Q_t\le (8/\epsilon)M$ imply $L_t\ge H-(8/\epsilon)M$ which is at least $(c^2-1)\lambda^2M$ using Observation~\ref{obs:trivial}. We split the first term above as follows
\begin{eqnarray}
 \Pr[\exists t: L_t \ge (c^2-1)\lambda^2M \cap F ] &\le&  \Pr[\exists t: L_t \ge (c^2-1)\lambda^2M \cap G_t ]  
 +\Pr[\exists t:\overline{G_t}\cap F ]  \label{eqncombi1}
 \end{eqnarray}
Let us first look at the second term in the right hand side above. Since the event $\overline{G_t}$ means that the squared $\ell_2$ discrepancy must have become huge at some time prior to $t$, we get
\begin{eqnarray} 
\Pr[\exists t:  \overline{G_t}\cap F ]
& =& \Pr[\{\exists t: \exists t'<t: D(t')^2\ge H \}\cap F ] \nonumber\\
& =& \Pr[\{\exists t: D(t)^2\ge H\}\cap F]  \label{eqncombi2}
\end{eqnarray}

Conditioned on $F$ occuring, let $t_s$ be the smallest time for which it holds that $D(t_s)^2\ge H$. But this implies $D(t')^2 < H$ for all $t'<t_s$, i.e. $G_{t_s}$ occurs. Also $D(t_s)^2\ge H$ and $F$ together imply $L_{t_s}\ge (c^2-1)\lambda^2M$. So in fact for time $t_s$, both the events $L_{t_s} \ge (c^2-1)\lambda^2M$ and $G_{t_s}$ must occur. Hence
\begin{equation}
  \Pr[\{\exists t: D(t)^2\ge H\}\cap F]\le \Pr[\exists t: L_t \ge (c^2-1)\lambda^2M \cap G_t ].  \label{eqncombi3}
  \end{equation}
Combining equations (\ref{end3}), (\ref{eqncombi1}), (\ref{eqncombi2}) and (\ref{eqncombi3}) we see that
\begin{eqnarray*}
\Pr[\exists t: D(t)^2\ge H]  
&\le& 2 \Pr[\exists t: L_t \ge (c^2-1)\lambda^2M\cap G_t ] +Pr[\overline{F} ] \\
&\le & 3e^{-\epsilon\lambda^2/64} \qquad\textrm{(using Claims \ref{claim:sketchlinear} and \ref{claim:sketchquad}) }
\end{eqnarray*}
\end{proof}

It only remains to prove Claims \ref{claim:sketchlinear} and \ref{claim:sketchquad}. We prove them now.

\begin{proof}(of Claim~\ref{claim:sketchlinear})
It is easy to see that $L_t$ is a martingale. We bound it similar to the $\ell_\infty$ discrepancy. For an element $i$ and row $j$, let $t_{ij}$ denote the last time $i$ was not included in $C_{j,S}(t)$ i.e. $t_{ij}+1$ is the first time when element $i$ becomes corrupt for $(j,S)$. 
Define the random variable
\[Y_t=L_t-\eta  \sum_j\sum_{i\in C_{j,S}(t)} a_{ji}^2\left[x_t(i)^2 -x_{t_{ij}}(i)^2\right].\]
where $\eta$ is a parameter whose exact value will be determined later. Then, 
\[ Y_t-Y_{t-1} = 2\gamma \left\langle r_t,\sum_j disc_{t-1}(j,S)\sum_{i\in C_j^t}a_{ji}u_i^t -\eta  \sum_j\sum_{i\in C_{j,S}(t)}a_{ji}^2x_{t-1}(i)u_i^t\right\rangle -\eta\gamma^2  \sum_j\sum_{i\in C_{j,S}(t)}a_{ji}^2\langle r_t,u_i^t\rangle^2  .\]
Using Lemma~\ref{lem:randvector}, this gives
 \[ \E_{t-1}[Y_t-Y_{t-1}]=-\eta\gamma^2 \sum_j\sum_{i\in C_{j,S}(t)}a_{ji}^2 \|u_i^t\|_2^2\]
and 
\begin{eqnarray*}
&&\E_{t-1}[(Y_t-Y_{t-1})^2] =   O(n^8m^4\gamma^3) +4\gamma^2\|\sum_j  disc_{t-1}(j,S)\sum_{i\in C_{j}^t}a_{ji}u_i^t-\eta\sum_j\sum_{i\in C_{j,S}(t)}a_{ji}^2x_{t-1}(i)u_i^t\|_2^2 \\
& \le & O(n^8m^4\gamma^3) +(8/\epsilon)\gamma^2 \sum_{i} \left[\sum_{j:i\in C_{j}^t}  disc_{t-1}(j,S)a_{ji}-\eta\sum_{j:i\in C_{j,S}(t)} a_{ji}^2x_{t-1}(i)\right]^2\|u_i^t\|_2^2 
\end{eqnarray*}
where the inequality uses property (ii) of Universal Vector Coloring.
Conditioned on the event $G_t$, the above can be bounded as follows. First using Cauchy-Schwarz inequality, we observe that 
\[ (\sum_{j:i\in C_j^t}disc_{t-1}(j,S)a_{ji})^2\le D(t-1)^2  \sum_{j:i\in C_j^t} a_{ji}^2\le H  \sum_{j:i\in C_j^t} a_{ji}^2 .\]
Notice here that the variance depends on the $\ell_2$ discrepancy at the previous time step, and thus to ensure that the variance remains small, we needed to ensure that the $\ell_2$ discrepancy never becomes huge. For this reason, we work with a bound conditioned on $G_t$. 
Now for $\eta\le (c/2)\sqrt{H}$ and $c\ge 2$, the coefficient of $\|u_i^t\|_2^2$ in the variance can be bounded as 
\begin{eqnarray*}
 \left[\sum_{j:i\in C_j^t}  disc_{t-1}(j,S)a_{ji}-\eta\sum_{j:i\in C_{j,S}(t)} a_{ji}^2x_{t-1}(i)\right]^2 &\le & 2H  \sum_{j:i\in C_j^t} a_{ji}^2+2\eta^2(\sum_{j:i\in C_{j,S}(t)} a_{ji}^2)^2 \\
 &\le & (c^2/2)H  \sum_{j:i\in C_{j,S}(t)} a_{ji}^2 (1+ \sum_{j} a_{ji}^2) \\
 & \le & c^2H  \sum_{j:i\in C_{j,S}(t)} a_{ji}^2 \qquad\textrm{(using $\|B\|_2\le 1$)}
 \end{eqnarray*}
Thus (ignoring the $O(n^8m^4\gamma^3)$ term like before) 
\[ \E_{t-1}[(Y_t-Y_{t-1})^2] \le  (8/\epsilon)c^2\gamma^2 H \sum_{i}\sum_{j:i\in C_{j,S}(t)}a_{ji}^2\|u_i^t\|_2^2= (8/\epsilon)c^2\gamma^2 H \sum_{j}\sum_{i\in C_{j,S}(t)}a_{ji}^2\|u_i^t\|_2^2 .\]
Using Corollary~\ref{cor:freedmanbetter} now with $\alpha=\epsilon\eta/8c^2H$, and noting that 
\[ Y_t\ge L(t)-\eta \sum_j\sum_{i\in C_{j,S}(t)} a_{ji}^2 \ge L(t)-\eta M\]
we get
\begin{eqnarray}
\Pr[\{\exists t: L_t\ge \left(\frac{c^2-1}{c^2}\right)H  \cap G_t\}] &\le& \Pr [\{\exists t: Y_t\ge \left(\frac{c^2-1}{c^2}\right)H -\eta M \cap G_t\}] \nonumber\\
& \le & \exp\left( \frac{-\epsilon\eta(\frac{c^2-1}{c^2}H-\eta M)}{8c^2H}\right) \label{eqnlindsey}
\end{eqnarray}

Above we used Corollary~\ref{cor:freedmanbetter} for a sequence of random variables conditioned on some event ($G_t$), and it needs to be justified why we can do that. The conditioning on $G_t$ is used to imply that the variance of the sequence $Y_t$ is small. We can then define another sequence $\tilde{Y}_t$ such that if the variance ever becomes large, increments in $\tilde{Y}_t$ are zero from then on. We can then apply Corollary~\ref{cor:freedmanbetter} on the sequence $\tilde{Y}_t$ to get the same result.

Getting back to the proof, 
(\ref{eqnlindsey}) is minimized at $\eta=\frac{c^2-1}{2c^2M}H$ and we get 
\begin{eqnarray*}
\Pr[\{\exists t: L_t\ge \left(\frac{c^2-1}{c^2}\right)H  \cap G_t\}]  \le \exp\left( -\frac{\epsilon(\frac{c^2-1}{c^2})^2H}{32c^2M}\right) \le \exp(-\epsilon\lambda^2/64)
\end{eqnarray*}
where the last inequality uses $H=c^2\lambda^2M$ and holds for $c\ge 2$. 
\end{proof}

\begin{proof}(of Claim~\ref{claim:sketchquad})
We want to show
\[\Pr [\exists t: Q_t>(8/\epsilon) M]\le e^{-\epsilon\lambda^2/64} .\] 
This follows similarly to the rest of the proof and we show only the main outline here. Similar to how we bound the vaue of $L_t$, define a random variable
\[Y_t=Q_t-\eta \sum_j \sum_{i\in C_{j,S}(t)}a_{ji}^2[x_t(i)^2-x_{t_{ij}}(i)^2] .\]
where $\eta\ge 4/\epsilon$ is a constant. It can be computed that
\[ \E_{t-1}[Y_t-Y_{t-1}]\le (\frac{2}{\epsilon}-\eta)\gamma^2\sum_j\sum_{i\in C_{j,S}(t)}a_{ji}^2\|u_i^t\|_2^2 \le -(\eta/2)\gamma^2\sum_j\sum_{i\in C_{j,S}(t)}a_{ji}^2\|u_i^t\|_2^2 \]
where the last inequality holds for $\eta\ge 4/\epsilon$, and
\begin{eqnarray*}
 &&\E_{t-1}[(Y_t-Y_{t-1})^2] \le O(n^8m^2\gamma^3)+4\gamma^2\eta^2 \|\sum_j \sum_{i\in C_{j,S}(t)}a_{ji}^2x_{t-1}(i)u_i^t\|_2^2 \\
 &\le &O(n^8m^2\gamma^3)+(8/\epsilon)\gamma^2\eta^2 \sum_i \left[\sum_{j:i\in C_{j,S}(t)}a_{ji}^2x_{t-1}(i)\right]^2\|u_i^t\|_2^2 \qquad\textrm{(using Lemma~\ref{lem:randvector})}\\
  &\le &O(n^8m^2\gamma^3)+(8/\epsilon)\gamma^2\eta^2 \sum_i \sum_{j:i\in C_{j,S}(t)}a_{ji}^2\|u_i^t\|_2^2 \qquad\textrm{(using $\|B\|_2\le 1$ and $|x_{t-1}(i)|\le 1$)} \\
  &=&O(n^8m^2\gamma^3)+(8/\epsilon)\gamma^2\eta^2 \sum_j \sum_{i\in C_{j,S}(t)}a_{ji}^2\|u_i^t\|_2^2
 \end{eqnarray*}
 We can again ignore the $O(n^8m^2\gamma^3)$ term as before. Applying Corollary~\ref{cor:freedmanbetter} now with $\alpha=\frac{\epsilon}{16\eta}$, and noting that 
\[ Y_t\ge Q(t)-\eta \sum_j\sum_{i\in C_{j,S}(t)} a_{ji}^2 \ge Q(t)-\eta M \]
we get
\begin{eqnarray*}
\Pr[\exists t: Q_t> (8/\epsilon)M] &\le& \Pr [\exists t: Y_t\ge (8/\epsilon -\eta) M ] \\
& \le & \exp\left( -\frac{\epsilon(8/\epsilon-\eta)M}{16\eta}\right) \le e^{-\epsilon\lambda^2/16}\le e^{-\epsilon\lambda^2/64}
\end{eqnarray*}
for $\eta=4/\epsilon$. 

\end{proof}

\section*{Appendix B: Discrepancy of Polytopes} 
\label{sec:appc}

\begin{proof}(of Theorem~\ref{thm:polytopes})
Recall the proof of Theorem~\ref{thm:tusnady}. We constructed a set of canonical boxes, and then every axis-parallel box was written as a disjoint union of canonical boxes along with $O_d(\log^{2d-2}n)$ points. This then combined with Theorem~\ref{thm:vecgenunified} gave a discrepancy of $O_d(\log^dn)$.

A similar proof also extends to bound the discrepancy of $POL(H)$. The only question is how to make a set of ``canonical shapes" such that every polytope in $POL(H)$ can be written as a disjoint union of these canonical shapes plus some small set of points. This was already done in \cite{Mat99} and we provide a sketch here.

Let's first look at $d=2$ and further simplify the problem by assuming $k=2$. Let $H=\{\ell_1,\ell_2\}$. If $\ell_1$ and $\ell_2$ are perpendicular, then this is the same as Tusnady's problem. So assume that the two lines are not perpendicular and we want to bound the discrepancy of all parallelepiped generated by these two lines. Without loss of generality, assume $\ell_1$ points along the $x$-axis.
Then instead of making canonical rectangles, we construct canonical parallelepipeds, aligned along the two lines in $H$. We first make canonical $x$-intervals exactly as we do for constructing canonical boxes. In each canonical $x$-interval, we now look at the intervals constructed by translates of the line $\ell_2$, whereas for canonical boxes we made such intervals horizontally. The rest of the proof remains the same.

For a general $k$, we just repeat the above construction for every pair of lines in $H$. This only affects the discrepancy by polynomial factors in $k$.  

The above approach can then be easily extended to $\mathbb{R}^d$ also. 
\end{proof}

\section*{Appendix C: Canonical Boxes in $\mathbb{R}^d$}
\label{sec:appbox}
\begin{proof}(of Lemma~\ref{lem:canonical})
We will refer to the coordinate axes by $x_1,\dots,x_d$.
Let $P$ be a set of $n$ points and let $\ell=16[d+\log_2n]^{d-1}$. Sort the points in $P$ in increasing value of $x_1$-coordinate and let this ordering be given by $p_1,\dots,p_{n}$. Define a canonical $x_1$-interval to be an interval containing the points $p_{k2^q+1},\dots,p_{(k+1)2^{q}}$ for $k\ge 0$ and $q\ge \log_2 \ell$. 
For each canonical $x_1$-interval, say $t_{1}$, do a similar construction along the $x_2$ axis for the points contained in $t_1$ to get canonical intervals along $x_2$ axis. Let the set of these canonical $x_2$ intervals made within $t_1$ be $T_2|t_1$. A canonical $(x_1,x_2)$-interval is defined as the intersection of a canonical $x_1$ interval, say $t_1$, with a canonical $x_2$-interval in $T_2|t_1$.
Proceed similarly by constructing canonical $x_i$ intervals within each canonical $(x_1,x_2,\dots,x_{i-1})$-interval. Finally a canonical $(x_1,\dots,x_d)$-interval will be called as a canonical box. Denote the set of all canonical boxes constructed in this manner for $P$ by $\mathcal{B}(P)$.

We first bound the size of $\mathcal{B}(P)$. Call a canonical $(x_1,\dots,x_k)$-interval to be of type $(i_1,\dots,i_k)$ if it was constructed by first making a canonical $x_1$-interval of size $2^{i_1}\ell$ and then inside that a canonical $(x_1,x_2)$-interval of size $2^{i_2}\ell$ and so on where $i_1\ge i_2\ge \dots\ge i_k\ge 0$  and $1\le k\le d$.
We see that along $x_1$-axis, there are $\frac{n'}{2^{i_1}\ell}$ canonical $x_1$-intervals of size $2^{i_1}\ell$. For each of these, there are $\frac{2^{i_1}}{2^{i_2}}$ canonical $(x_1,x_2)$-interval of size $2^{i_2}\ell$. Continuing in this way, the number of canonical boxes of type $(i_1,\dots,i_d)$ is at most
\[ \frac{n}{2^{i_1}\ell}.\frac{2^{i_1}}{2^{i_2}}\dots\frac{2^{i_{d-1}}}{2^{i_{d}}}=\frac{n}{2^{i_d}\ell}.\]
So in total, the number of canonical boxes in $\mathcal{B}(P)$ is at most
\[ \sum\limits_{\log_2(n/\ell)\ge i_1\ge i_2\ge \dots \ge i_d\ge 0}\frac{n}{2^{i_d}\ell}\le \sum_{i_d\ge 0}\frac{n}{2^{i_d}\ell}[d+\log_2(n/\ell)]^{d-1}\le \frac{2n}{\ell}[d+\log_2(n/\ell)]^{d-1} \le \frac{n}{8}.\]


Now it remains to be shown that any axis-parallel box $R$ can be written as a disjoint union of some canonical boxes from $\mathcal{B}(P)$ along with $O_d(\log^{2d-2} n)$ leftover points. 
To see this, first observe that we can decompose $R$ into $O(\log n)$ disjoint strips along $x_1$, where the first and the last strip contains at most $\ell$ points and the other strips are canonical $x_1$-intervals. Each strip which is a canonical $x_1$-interval can again be decomposed in $O(\log n)$ intervals along $x_2$ which, expect the first and last interval, correspond to canonical $(x_1,x_2)$-intervals. The first and last intervals contain at most $\ell$ points. Proceeding this way, $R$ can be written as a disjoint union of some canonical boxes along with some leftover points. The number of these leftover points can be bounded by
\[ O_d( 2\ell+ \log n[2\ell+\log n(\dots)]\] 
where the sum continues until $d$ appearances of $\ell$ and thus equals
\[  O_d(\log^{d-1}n+\log^{d}n+\dots+\log^{2d-2}n) =O_d(\log^{2d-2}n).\]
\end{proof}

\end{document}